\documentclass[12pt,reqno,english]{article}
\usepackage[LGR,T1]{fontenc}
\usepackage[latin9]{inputenc}
\usepackage[letterpaper]{geometry}
\usepackage{color}
\usepackage{babel}
\usepackage{amsmath}
\usepackage{amsthm}
\usepackage{amssymb}
\usepackage{setspace}
\usepackage{nicefrac}
\usepackage[authoryear,longnamesfirst]{natbib}
\usepackage[unicode=true,pdfusetitle,
 bookmarks=true,bookmarksnumbered=false,bookmarksopen=false,
 breaklinks=false,pdfborder={0 0 1},backref=false,colorlinks=true]
 {hyperref}
\hypersetup{
 pdfborderstyle=,linkcolor=magenta,citecolor=blue}

\makeatletter

\ProvideTextCommand{\~}{LGR}[1]{\char126#1}

\theoremstyle{plain}
\newtheorem{thm}{\protect\theoremname}

\usepackage{caption}
\usepackage{subcaption}
\usepackage{bm}

\usepackage{color}
\usepackage{amsthm}
\usepackage[plain]{fullpage}     

\usepackage{graphicx}
\graphicspath{ {./graphs/} }
\usepackage{babel}

\usepackage{array}
\usepackage{booktabs}
\usepackage{multirow}

\numberwithin{equation}{section}

\definecolor{cornellred}{RGB}{179,27,27} 
\definecolor{cornellblue}{RGB}{55,135,176}
\definecolor{cornellgrey}{RGB}{96,94,92}
\usepackage{tikz}

\newcommand{\priorlocation}{\bar{\mu}}
\newcommand{\wnorm}{\lVert w \lVert_{\Sigma} }
\newcommand{\wdotnorm}{\lVert \Sigma^{-1} \dot{\mu}\lVert_{\Sigma}}

\theoremstyle{remark}
\theoremstyle{definition}
\newtheorem{definition}{\protect\definitionname}\theoremstyle{plain}
\theoremstyle{plain}
\theoremstyle{definition}
\newtheorem{example}{\protect\examplename}
\theoremstyle{definition}
\newtheorem{assumption}{Assumption}

\theoremstyle{plain}
\newtheorem{corollary}{Corollary}

\newtheorem{lemma}{\protect\lemmaname}[section]

\usepackage{graphicx}
\usepackage{psfrag}\usepackage{epsf}\usepackage{enumerate}
\usepackage{url}

\usepackage{amsthm}
\usepackage{amsfonts}

\usepackage{pdflscape}
\usepackage{threeparttable}

\providecommand{\conditionname}{Condition}
\providecommand{\definitionname}{Definition}
\providecommand{\examplename}{Example}
\providecommand{\lemmaname}{Lemma}
\providecommand{\propositionname}{Proposition}
\providecommand{\remarkname}{Remark}
\providecommand{\theoremname}{Theorem}

\usepackage{verbatim}

\makeatother

\providecommand{\theoremname}{Theorem}

\begin{document}
\sloppy
\title{Robust Bayes Treatment Choice with Partial Identification\thanks{We would like to thank the co-editor and three anonymous referees,  Karun Adusumilli, Tim Armstrong,  Tim Christensen, Patrik Guggenberger, Kei Hirano, Toru Kitagawa and seminar audiences at  Penn State for helpful feedback, comments, and suggestions. We gratefully acknowledge financial support from NSF grant SES-2315600.}}
\author{Andr\'es Aradillas Fern\'andez\thanks{Department of Economics, Massachusetts Institute of Technology, aradilla@mit.edu}\and Jos\'e Luis Montiel Olea\thanks{Department of Economics, Cornell University, jlo67@cornell.edu} \and Chen Qiu\thanks{Department of Economics, Cornell University, cq62@cornell.edu} \and J\"org Stoye\thanks{Department of Economics, Cornell University, stoye@cornell.edu} \and Serdil Tinda\thanks{Department of Economics, Stanford University, serdil@stanford.edu }}
\date{December 2025}

\maketitle

\vspace{-0.5cm}

\begin{abstract}
We study a class of binary treatment choice problems with partial identification through the lens of robust (multiple prior) Bayesian analysis. We use a convenient set of prior distributions to derive ex-ante and ex-post robust Bayes decision rules, both for decision makers who can randomize and for decision makers who cannot.

Our main messages are as follows: First, ex-ante and ex-post robust Bayes decision rules do not agree in general, whether or not randomized rules are allowed. Second, randomized treatment assignment for some data realizations can be optimal in both ex-ante and, perhaps more surprisingly, ex-post problems. Therefore, it is usually with loss of generality to exclude randomized rules from consideration, even when regret is evaluated ex post. 

We apply our results to a stylized problem where a policy maker uses experimental data to choose whether to implement a new policy in a population of interest, but is concerned about the external validity of the experiment at hand \citep{stoye2012minimax}, and to aggregation of data generated by multiple randomized control trials in different sites to make a policy choice in a population for which no experimental data are available \citep{manski2020towards,ishihara2021}.

\end{abstract}
\newpage{}

\onehalfspacing

\section{Introduction}

A policy maker must decide between implementing a new policy or preserving the status quo. Her data provide information about the potential benefits of these two options. Unfortunately, these data only \emph{partially identify} payoff-relevant parameters and may therefore not reveal, even in large samples, the correct course of action. Such \emph{treatment choice problems with partial identification} have recently received growing interest; for example, see \cite{d2021policy}, \cite{ishihara2021}, \cite{yata2021}, \cite*{christensen2022optimal}, \cite{kido2022distributionally} or \cite{manski2022identification}. Several interesting problems that arise in empirical research can be recast  using this framework. See \cite{MQS} and the end of this section for references.

This paper applies the \emph{robust Bayes} approach---which interpolates between Bayesian and agnostic minimax analyses by evaluating minimax risk over a set of priors---to a class of treatment choice problems with partial identification. The use of the robust Bayes approach has drawn recent attention in problems that feature partial identification \citep{GiacominiKitagawa,GKR,christensen2022optimal}. Indeed, due to partial identification, integrating Bayesian and minimax elements into  decision making can be particularly attractive (see \citealt{Poirier,MoonSchorfheide} and references therein). 

The robust Bayes approach can be applied \emph{ex-ante} or \emph{ex-post}, depending on whether the multiple priors are used to evaluate payoffs before or after seeing the data.\footnote{\cite{GKR} discuss both notions; they refer to the ex-ante and ex-post problems as ``Gamma-minimax'' and ``Conditional Gamma-minimax'', respectively. \cite{christensen2022optimal} focus on the ex-post problem for treatment choice problems with partial identification in a restricted class of decision rules.} These concepts represent two different ways of resolving model ambiguity and sampling uncertainty, and both have been proposed to improve Bayesian robustness in decision problems.  

By the well-known \emph{dynamic consistency} of Bayesian decision making, ex-ante and ex-post robust Bayes coincide with each other---and with standard Bayes optimality---if the set of priors is a singleton. Indeed, this equivalence is used to calculate Bayes optimal decisions in practice because ex-post rules are usually computed but ex-ante Bayes optimality is claimed. As pointed out for the present context by \cite{GKR}, they do not in general agree otherwise.\footnote{For an estimation problem with a quadratic loss, \citet[][Appendix B]{kitagawa2012estimation} derives the ex-post $\Gamma$-minimax estimator and shows that it is not ex-ante $\Gamma$-minimax optimal.} 
However, to what extent this inequivalence affects treatment choice problem with partial identification is far from clear. The class of priors that we consider has a Cartesian product structure resembling \textit{rectangularity}, a condition under which maximin welfare loss is known to be dynamically consistent.\footnote{See \citet{epstein2003recursive} and also \citet{wakai2007note}, \citet{amarante2019recursive}, and references therein.} While this result does not apply here---the priors are not rectangular in the strict technical sense, and the theoretical results were not established for regret loss---one might wonder if the conclusion holds anyway or else, what qualitative and quantitative relationships exist between ex-ante and ex-post robust Bayes in the presence of partial identification. Using a convenient class of priors allows us to exhaustively answer these questions in a case that we believe holds some interest. 

To this end, we use the same framework as \citet{yata2021} and \citet{MQS} but impose a simple instance of \citeauthor{GiacominiKitagawa}'s (\citeyear{GiacominiKitagawa}) set of priors, namely a symmetric and uniform two-point prior for reduced-form parameters and no restriction at all on unidentified parameters given reduced-form parameters. Working with regret, we formally define ex-ante and ex-post robust Bayes and, following \citet{berger1985statistical}, label them as ``$\Gamma$-minimax regret'' ($\Gamma$-MMR) and ``$\Gamma$-posterior expected regret'' ($\Gamma$-PER).\footnote{Among others, see \citet{Savage51}, \citet{manski2004statistical}, \citet{stoye2012new}, and \citet{MQS} for justifications of focusing on regret in treatment choice problems. In particular, while minimax loss can be an attractive alternative to minimax regret, it leads to trivial recommendations in treatment choice settings including our examples.}  We then precisely characterize when these notions coincide and when they disagree. The main qualitative insights are as follows:
\begin{itemize}
    \item Ex-ante $\Gamma$-MMR and ex-post $\Gamma$-PER criteria do frequently disagree. This conclusion does not hinge on whether or not randomization rules are allowed. If randomization is permitted, the criteria differ whenever updating the prior for the reduced-form parameter does not (ex-ante almost surely) resolve ambiguity regarding the sign of the optimal action. Even if randomization is not allowed, the criteria will disagree whenever the \emph{identification power} of the model (in a sense we make precise) is sufficiently small compared to the \emph{informativeness of the data}.
    \item Randomized rules are general solution concepts in both $\Gamma$-MMR and $\Gamma$-PER problems. That is, optimal rules often can or even must randomize even if regret is evaluated ex-post. For the ex-ante $\Gamma$-MMR problem, whenever the identification power of the model is sufficiently small compared to the informativeness of the data, (1) infinitely many optimal rules exist, and (2) many of them are randomized. For the ex-post $\Gamma$-PER criterion, one might conjecture by analogy to single-prior Bayes inference that the optimal rules do not randomize. However, as long as (1) there exists ambiguity regarding the sign of the optimal action at the location of the two-point prior for the reduced-form parameter and (2) randomized rules are allowed, the optimal rule is unique and is randomized. Hence, exclusion of randomized rules is with loss of generality and cannot be justified by simply evaluating regret ex-post---it must come from other considerations, be they second-order preferences or logistical or convenience concerns. 
\end{itemize}

Our point is not to advocate for either notion of robust Bayes criterion. We also do not aim to solve for robust Bayes criteria for more general sets of priors as this would get much more involved but (we suspect) not much more instructive to illustrate the points discussed above. What we hope to illustrate is when, and how, $\Gamma$-MMR and $\Gamma$-PER criteria differ. We also relate these results to timing assumptions in a fictitious game between the policy maker and an adversarial Nature.

Several auxiliary findings might be of independent interest. First, for $\Gamma$-MMR, even if we restrict the set of decision rules to be a class of non-randomized threshold rules based on the ``efficient'' linear index, the optimal threshold is not always zero. Given the apparent symmetry of the problem, we find this feature rather surprising. Second, whenever the dimension of the signal is larger than $1$, there always exist (regardless of the parameter space and the variance of the signals) non-randomized linear-index threshold rules (with a threshold equal to zero) that are $\Gamma$-MMR optimal (among all decision rules). This is in stark contrast to \cite{MQS}, in which no linear index rule is globally minimax regret optimal if the degree of partial identification is severe. The intuition is that the prior much reduces the state space; the signal space then becomes so rich relative to the state space that even linear threshold rules can effectively mimic randomization.

The literature on treatment choice with partially identified parameter has been growing  since \cite{manski2004statistical} and \cite{Dehejia2005}. For partial identification with known distribution of data, \cite{Manski2000,manski2005social,manski2007identification} and \citet{Stoye07} find minimax regret optimal treatment rules. For finite-sample minimax regret results with model ambiguity and sampling uncertainty, see \cite{stoye2012minimax,stoye2012new}, \cite{yata2021}, \cite{ishihara2021} and \cite{MQS}; \citeauthor{kido2023locally}'s (\citeyear{kido2023locally}) analysis is asymptotic. Bayes and robust Bayes approaches are analyzed by \citet{chamberlain2012}, \citet{GiacominiKitagawa}, \citet{GKR}, \citet{christensen2022optimal}, among others. Earlier investigations of ex-ante and ex-post $\Gamma$-minimax estimators include \citet{dasgupta1989frequentist} and \citet{betro1992conditional}. See \cite{vidakovic2000gamma} for a review.  For different settings with point-identified welfare, finite- and large-sample results on optimal treatment choice rules were derived by \cite{canner}, \citet{ChenGug}, \citet{HiranoPorter2009,HiranoPorter2020}, \citet*{kitagawa2022treatment},  \citet{schlag2006eleven}, \citet{stoye2009minimax}, and \citet{tetenov2012statistical}. There is also a large literature on optimal policy learning with covariates containing results with point identified \citep{BhattacharyaDupas2012,kitagawa2018should,KT19,MT17,KW20,AW20,KST21,ida2022choosing} as well as partially identified \citep{kallus2018confounding,ben2021safe,ben2022policy,d2021policy, christensen2022optimal,adjaho2022externally,kido2022distributionally,lei2023policy} parameters. \citet{GugQuantile} and \citet{ManskiTetenovJER} analyze related problems but focus on quantile, as opposed to expected, loss; \citet{song} considers partial identification but mean squared error regret loss.

The rest of this paper is organized as follows. Section \ref{sec:frame} sets up the problem, provides examples, and defines both versions of robust Bayes optimality. Section \ref{sec:main} contains complete solutions for all aforementioned scenarios and relates them to timing assumptions in the ``Games against Nature'' interpretation of minimax theory. Section \ref{sec:conclusion} concludes. Proofs and auxiliary results are collected in the Appendix.

\section{Framework}\label{sec:frame}

\subsection{Actions, Payoffs, Statistical Model, and Decisions}

Our setup follows \cite{MQS}, who in turn follow \cite{Ferguson67} and others. Consider a policy maker  who needs to choose an \emph{action} $a\in[0,1]$ interpreted as probability of assigning treatment in the target population.\footnote{Randomization could be i.i.d. across future potential treatment recipients, \emph{fractional} in the sense of randomly assigning a certain fraction of the treatment population (in this sense, $a\in[0,1]$ can also be interpreted as the fraction of the population receiving the treatment), or an ``all or nothing'' randomization for the entire treatment population. While these might not be practically equivalent in all applications, they are in the current decision theoretic framework. See \cite{manski2007admissible} for an exception in the related literature.} Her payoff when taking action $a\in[0,1]$ is captured by the welfare function \begin{equation}\label{eq:welfare}
W(a,\theta):=aW(1,\theta)+(1-a)W(0,\theta),    
\end{equation}
where $\theta\in\Theta$ is an unknown state of the world or \emph{parameter} and the functions $W(1,\cdotp):\Theta\rightarrow\mathbb{R}$ and $W(0,\cdotp):\Theta\rightarrow\mathbb{R}$ are known. Here, we may interpret $W(1,\theta)$ and  $W(0,\theta)$ as the welfare of actions $a=1$ (treating everyone in the population) and action $a=0$ (treating no one in the population). Therefore, \eqref{eq:welfare} implies that welfare is linear in actions, a standard assumption in the literature. Denote by $U(\theta):=W(1,\theta)-W(0,\theta)$ the \emph{welfare contrast} at $\theta$. If $U(\theta)$ were known to the policy maker, her optimal action would simply be 
\begin{equation}\label{eq:infeasible}
\mathbf{1}\left\{ U(\theta)\geq0\right\}.
\end{equation}
The policy maker does not know $U(\theta)$ but can learn about $\theta$. Specifically, we assume that she observes a random vector $Y\in\mathbb{R}^{n}$ with multivariate normal distribution 
\begin{equation}\label{eq:normal_model}
Y\sim N(m(\theta),\Sigma),
\end{equation}
where the function $m(\cdotp):\Theta\rightarrow\mathbb{R}^{n}$ and the positive definite matrix $\Sigma$ are known.

Our focus is on the case when the data is not entirely informative about the sign of $U(\theta)$: Even if the policy maker perfectly learned $m(\theta)$, she could not (necessarily) pin down the sign of $U(\theta)$. To formally model such \emph{treatment choice problems with (decision-relevant) partial identification}, let
\begin{equation}
M:=\left\{ \mu\in\mathbb{R}^{n}:m(\theta)=\mu,\theta\in\Theta\right\} \label{eq:M}
\end{equation}
collect all the means of $Y$ that can be generated as $\theta$ ranges over $\Theta$.
We refer to elements $\mu\in M$ as $\emph{reduced-form}$ parameters
because they are identified in the statistical model \eqref{eq:normal_model} without further assumptions.
Define the \emph{identified set} for the welfare contrast given $\mu$ as 
\begin{equation}
I(\mu):=\left\{ u\in\mathbb{R}:U(\theta)=u,m(\theta)=\mu,\theta\in\Theta\right\} \label{eq:I_mu}
\end{equation}
and the corresponding upper and lower bounds as 
\begin{equation}
\overline{I}(\mu):=\sup I(\mu),\quad\underline{I}(\mu):=\inf I(\mu).\label{eq:I_bound}
\end{equation}

Henceforth, when we refer to a  treatment choice problem with partial identification, we mean there exists some nonempty open set in $M$ such that for all $\mu$ in that open set, $\underline{I}(\mu)  <0<\overline{I}(\mu)$. For simplicity, we also assume that the infimum and supremum in \eqref{eq:I_mu} are attained.

A \emph{decision rule} $d:\mathbb{R}^{n}\rightarrow[0,1]$ is a (measurable) mapping from data $Y$ to the unit interval $[0,1]$. We call $d$ \emph{non-randomized} if it (almost surely, a.s.) maps into $\{0,1\}$; otherwise, we call $d$ \emph{randomized}, including if it randomizes for some but not all data realizations. We use $\mathcal{D}_n$ to denote the set of all decision rules, and we consider decision rules the same if they a.s. agree. As a result, a rule is $\emph{unique}$ only up to a.s. agreement. The \emph{oracle policy} $\textbf{1}\{U(\theta) \geq 0\}$ is of special interest and for any given $\theta$ is contained in $\mathcal{D}_n$, but is not feasible in the statistical sense because $U(\theta)$ is not known.  
In general, there will not be an unambiguously best feasible decision rule, a problem that gave rise to a large literature on different optimality criteria and their implementation. Before introducing the robust Bayes approach, we give two examples that fit into our general framework.

\subsection{Examples}

\begin{example}[\citeauthor{stoye2012minimax}, \citeyear{stoye2012minimax}]\label{ex:stoye}

The one-dimensional version of the general setup and has been frequently analyzed before \citep{manski2000identification,brock2006profiling,stoye2012minimax,tetenov2009measuring,kitagawa2023treatment}. One motivation for it is to think of a policy maker who uses experimental data to choose whether to implement a new policy in a population of interest, but is concerned about the external validity of the experiment at hand. The treatment effect of action $a=1$ is $\mu^*\in\mathbb{R}$, while the effect of action $a=0$ is normalized to $0$; thus, the policy maker's expected payoff equals $W(a,\mu^*):= a \cdot \mu^*$. The policy maker  observes a realization of the one-dimensional statistic
\begin{equation} \label{eq:stoye_data}
\hat{\mu} \sim N(\mu,\sigma^2),
\end{equation}
where $\sigma>0$ is known and where $\mu\in\mathbb{R}$ is an \emph{identifiable treatment effect}, i.e. it could be perfectly learned from infinite data. In this example, $\theta=(\mu,\mu^*)^{\top}$,  $\Theta\subseteq \mathbb{R}^2$, $m(\theta)=\mu$, and  $U(\theta)=\mu^*$.

Since the target population and the population from which data \eqref{eq:stoye_data} is collected can be different, partial identification naturally arises. We assume that the identifiable and true treatment effects are constrained by
$\left\vert \mu^*-\mu \right\vert \leq k$
for some known $k \geq 0$, implying
\[
I(\mu)=[\mu-k,\mu+k],\quad \overline{I}(\mu)=\mu+k,\quad \underline{I}(\mu)=\mu-k,\quad\forall \mu \in \mathbb{R}.
\]
The planner must choose a statistical decision rule $d\in\mathcal{D}_1:\mathbb{R}\rightarrow[0,1]$ that maps observed data $\hat{\mu}$ to an action $a\in[0,1]$. 
Taking \eqref{eq:stoye_data}  as an approximation, this stylized example could reflect model uncertainty (e.g., a treatment effect is estimated in a possibly somewhat misspecified model), external validity concerns (e.g., a randomized clinical trial was performed on volunteers), or a shift in the environment (e.g., we transfer estimates from study populations to treatment populations with slightly different covariates or are concerned about distributional drift over time). \qed
\end{example}

\begin{example}[\citeauthor{ishihara2021}, \citeyear{ishihara2021}]\label{ex:ik}

This example is taken from \citeauthor{ishihara2021}'s (\citeyear{ishihara2021}; see also \citet{manski2020towards}) ``evidence aggregation'' framework. A policy maker is interested in implementing a new policy in country $i=0$ and observes estimates of the policy's effect for countries $i=1,...,n$. Let $Y=(Y_1,...,Y_n)^\top \in \mathbb{R}^n$ denote these estimates and let $(x_0,\ldots,x_n)$ be nonrandom, $d$-dimensional baseline covariates. The policy maker is willing to extrapolate from her data by assuming that the welfare contrast of interest equals $U(\theta)=\theta(x_0)$ and that 
\begin{equation*}
Y = 
\begin{pmatrix}
Y_1 \\
\vdots \\
Y_n
\end{pmatrix}
\sim N(m(\theta), \Sigma), \quad m(\theta) = 
\begin{pmatrix}
\theta(x_1) \\
\vdots \\
\theta(x_n)
\end{pmatrix}, 
\quad
\Sigma =  
\operatorname{diag}(\sigma_1^2,\ldots,\sigma_n^2),
\end{equation*} 
where $\theta: \mathbb{R}^d \rightarrow \mathbb{R}$ is an unknown Lipschitz function with known constant $C$. For notational simplicity, we can further write $\mu_i$ for $\theta(x_i)$. Thus, $Y_i \sim N(\mu_i,\sigma_i^2)$, $Y\sim N(\mu,\Sigma)$ and $U(\theta)=\mu_0$. The policy question is: Given data $Y$, what proportion of the population in country $i=0$ should be assigned the new policy?

Let $\left\Vert \beta\right\Vert:=\sqrt{\beta^{\top}\beta}$ be the Euclidean norm of a vector $\beta$.
In this example, the identified set for the welfare contrast $\mu_0$ is 
\begin{equation*}
I(\mu) = \{ u \in \mathbb{R} :  \left\vert \mu_i-u \right\vert \leq C\left\Vert x_i-x_0\right\Vert, ~~ i=1,\ldots,n \}.
\end{equation*}
The lower and upper bounds on $I(\mu)$ are simply  intersection bounds:
\[ \underline{I}(\mu) = \max_{i=1,\ldots,n} \{ \mu_i - C\left\Vert x_i-x_0 \right\Vert  \},  \quad  \overline{I}(\mu) = \min_{i=1,\ldots,n} \left \{ \mu_i + C \left\Vert x_i-x_0 \right\Vert \right \}.
\]\qed
\end{example}

\subsection{Robust Bayes Optimality}\label{sec:bayes.optimality}

Our setting up to here is as in \citet{MQS}. We now connect it to the robust Bayes literature by imposing a \emph{set of priors} $\Gamma$ on $\theta$.  Following \cite{GiacominiKitagawa}, we choose a particular single proper prior $\pi_{\mu}$ for $\mu\in M$ but leave the conditional prior of $\theta$ given $\mu$, denoted as $\pi_{\theta \mid \mu}$, unrestricted except for
\begin{equation}\label{eq:conditional.prior}
\pi_{\theta \mid \mu}\{U(\theta)\in I(\mu)\}=1,\quad \pi_{\mu}\text{-a.s.}    
\end{equation}
Then, the class of priors $\Gamma$ consists of all priors on $\theta$ induced by the single prior $\pi_{\mu}$ and any conditional prior $\pi_{\theta \mid \mu}$ that meets \eqref{eq:conditional.prior}. Intuitively, we choose a single prior on the point-identified parameter and place no new restriction on the partially identified parameter $U(\theta)$. One can pick any proper prior $\pi_{\mu}$; for tractability, we let $\pi_{\mu}$ be supported on two symmetric points $\{\priorlocation,-\priorlocation\}$ with equal probability, where $\priorlocation\in\mathbb{R}^n$ is chosen by the decision maker. Henceforth, $\Gamma$ is understood to refer to the implied  set of priors: 
\begin{equation}\label{eq:Gamma}
\Gamma:=\left\{\pi_\theta=\int \pi_{\theta\mid \mu}d\pi_{\mu}: \pi_{\mu} \sim \text{unif}(\{-\priorlocation,\priorlocation\}),\pi_{\theta \mid \mu}\text{ satisfies \eqref{eq:conditional.prior}}\right\}. 
\end{equation}
 
While the set of priors $\Gamma$ broadly puts us into the ``robust Bayes'' territory, it still does not pin down a uniquely best decision rule because the sign of $U(\theta)$ can remain ambiguous. Let
\[
L(a,\theta):=\sup_{a^{\prime}\in[0,1]} W({a^{\prime},\theta}) - W(a,\theta)=U(\theta)\left\{ \mathbf{1}\{U(\theta)\geq0\}-a\right\}
\]
be the \emph{regret} of action $a\in[0,1]$. We evaluate decision rules $d$ by their \emph{expected regret}, defined as 
\begin{eqnarray}\label{eq:expected.regret}
R(d,\theta) &:=&  \mathbb{E}_{m(\theta)}[L(d(Y),\theta)] \notag \\
&=&  U(\theta)\left\{ \mathbf{1}\{U(\theta)\geq0\}-\mathbb{E}_{m(\theta)}[d(Y)]\right\},
\end{eqnarray}
where for any $x\in\mathbb{R}^n$, $\mathbb{E}_{x}[\cdot]$ denotes expectation taken over $Y\sim N(x,\Sigma)$.

Even for given set of priors $\Gamma$ given and commitment to expected regret, the Robust Bayes literature contains multiple decision criteria that do not in general agree. The difference lies in when (and how) the expectations regarding the unknown parameter $\theta\in\Theta$ are taken. For decision rule $d\in \mathcal{D}_n$, let
\[r(d,\pi):=\int_{\theta\in\Theta} R(d,\theta)d\pi(\theta)
\]
be its Bayes expected regret under a prior $\pi\in\Gamma$. Following \citet[][Definition 12, p. 216]{berger1985statistical}, we introduce the first robust Bayes optimality notion. 

\begin{definition}[\emph{Ex-ante} $\Gamma$-Minimax Regret]\label{def:MMR}
A decision rule $d^* \in \mathcal{D}_n$ is \emph{$\Gamma$-minimax regret} (henceforth $\Gamma$-MMR) optimal if
\begin{equation}\label{eq:gamma.mmr}
\sup_{\pi \in \Gamma}r(d^*,\pi)= \inf_{d \in \mathcal{D}_n} \sup_{\pi \in \Gamma} r(d,\pi).  
\end{equation}

\end{definition}

An alternative ``posterior'' robustness notion is also common in Bayesian analysis. For each action $a\in [0,1]$, define \emph{posterior expected regret} under prior $\pi$ \citep[][Definition 8, p. 159]{berger1985statistical}
\[
\rho(a,\pi_{\theta\mid Y}):=\int_{\tilde{\theta}\in\Theta}L(a,\tilde{\theta})d\pi_{\theta \mid Y}(\tilde{\theta}),
\]
where $\pi_{\theta\mid Y}$ is the posterior distribution of $\theta$ given prior $\pi$ and data $Y$.\footnote{In our setting, information from data $Y$ does not revise the conditional prior $\pi_{\theta\mid\mu}$ \citep{GiacominiKitagawa}. For any event $A$ in the $\sigma$-algebra of $\Theta$, we therefore have $\pi_{\theta\mid Y}(A)=\int\pi_{\theta\mid\mu}(A)d\pi_{\mu\mid Y}$, where $\pi_{\mu\mid Y}$ is the posterior distribution of $\mu$ given $Y$.} Then we have the following, alternative optimality criterion \citep[][Definition 10, p. 205]{berger1985statistical}:
\begin{definition}[\emph{Ex-post} $\Gamma$-Posterior Expected Regret]\label{def:PER}
A decision rule $d^* \in \mathcal{D}_n$ is \emph{$\Gamma$-posterior expected regret} (henceforth $\Gamma$-PER) optimal if
\begin{equation} \label{eq:gamma.per.2}
\sup_{\pi \in \Gamma}\rho(d^*(Y),\pi_{\theta\mid Y}) =\inf_{a \in [0,1]} \sup_{\pi \in \Gamma} \rho(a,\pi_{\theta\mid Y}),\quad \forall Y \in \mathbb{R}^n.
\end{equation}
If the decision maker's action space is restricted to $\{0,1\}$, i.e., randomization is not allowed, the above definitions are revised by replacing $[0,1]$ with $\{0,1\}$.
\end{definition}

The labeling of $\Gamma$-MMR as ``ex-ante'' versus $\Gamma$-PER as ``ex-post'' can be related to the timing of a fictitious game against an adversarial Nature; see Section \ref{sec:timing} for additional discussion. If $\Gamma$ were a singleton, the criteria would coincide and would also agree with (single-prior) Bayes optimality.
They do not in general agree otherwise. The term ``Gamma minimax'' usually (and even ``robust Bayes'' more often than not) refers to $\Gamma$-MMR; for example, see \citet[][Section 4.7.6]{berger1985statistical}.\footnote{\citet*{GKR} discuss both criteria for general loss functions and refer to Definitions \ref{def:MMR} and \ref{def:PER} as the ``unconditional $\Gamma$-minimax'' and ``conditional $\Gamma$-minimax'' problems, respectively. For treatment choice problems with partial identification, \cite{christensen2022optimal} optimize the $\Gamma$-PER criterion, restricting the action space to be $\{0,1\}$.}

While it is not our agenda to advocate for either criterion, some possible considerations are as follows. The ex-ante approach may be perceived as more suitable if the planner has commitment power and has also been justified axiomatically \citep{hayashi2008regret,stoye2011axioms}. We also find some numerical and theoretical evidence that the $\Gamma$-MMR rule may have desirable  \textit{frequentist properties}, e.g. in states of the world off the prior's support; see discussions at the end of Section \ref{sec:ex.post}. Regarding computational feasibility, the ex-post approach is often easier due to the applicability of backward induction and is routinely employed to quantify posterior robustness of statistical decisions. On the other hand, recent work on numerical discovery of minimax rules \citep{fernandez2024epsilon,guggenberger2025numerical} may render the ex-ante approach more scalable. See additional discussions in  Appendix \ref{sec:generality}.

The two-point structure of $\pi_\mu$ in \eqref{eq:Gamma} appears in several related treatment choice problems with minimax regret optimality criteria. For example, the least favorable prior takes such a form in a completely unconstrained minimax regret problem with point identified \citep{stoye2009minimax} or partially identified \citep{stoye2012minimax,yata2021,MQS} welfare contrast. In  analogously constrained minimax regret problems in which we restrict $\mu\in[-\left|\priorlocation\right|,\left|\priorlocation\right|]$, extending the analyses in the preceding literature would also imply a similar two-point symmetric structure for the least favorable prior. Given these precedents, we think our choice of $\pi_\mu$ represents a general feature of this class of  minimax regret decisions, in addition to offering computational tractability.

\section{Main Results}\label{sec:main}

In this section, we solve for the two versions of robust Bayes optimality under the set of priors \eqref{eq:Gamma}. Following \citet{yata2021}, we assume:
\begin{assumption}\label{asm:1}
\begin{itemize}
\item[(i)]  $\Theta$ is convex, centrosymmetric (i.e., $\theta\in\Theta$ implies $-\theta\in\Theta$) and nonempty.
\item[(ii)] $m(\cdot)$ and $U(\cdot)$ are linear.
\end{itemize}
\end{assumption}

These conditions are restrictive but encompass many examples of empirical relevance; see this paper's introduction, \citet{MQS}, and \citet{yata2021}. Exploiting symmetry of the setting, we also set
\[\overline{I}(\priorlocation)+\underline{I}(\priorlocation)>0.\]
This is a normalization because, by Lemma \ref{lemma:b1}, Assumption \ref{asm:1} implies
$\overline{I}(-\bar{\mu})+\underline{I}(-\bar{\mu})=-(\underline{I}(\bar{\mu})+\overline{I}(\bar{\mu}))$ and we could always replace $m(\cdot)\mapsto-m(\cdot)$; furthermore, the case of $\overline{I}(\bar{\mu})+\underline{I}(\bar{\mu})=0$ gives rise to trivial solutions.\footnote{In this case, the identified set for $U(\theta)$ does not change with true value of $\mu$, and under all decision criteria considered here, the solution will be a no-data rule that simply flips a coin.}

We say a rule is a \emph{linear-index threshold rule} if it has the form $\mathbf{1}\left\{ \beta^{\top}Y\geq c\right\} $
for some $\beta\in\mathbb{R}^{n}$ and $c\in\mathbb{R}$. Linear-index threshold rules are nonrandomized. They are of particular interest because they form a complete class when $U(\theta)$ is point-identified \citep{karlin1956theory} and have received particular attention in the recent literature \citep{ishihara2021,MQS}.  
For both $\Gamma$-MMR and $\Gamma$-PER, we  will clarify when linear-index threshold rules are optimal and when they are not. 
For reasons that will become obvious, the following linear-index threshold rule is of particular interest:
\begin{equation}\label{eq:efficient.linear.index}
d_{w,0}^{*}:=d_{w,0}^{*}(Y):=\mathbf{1}\left\{ w^{\top}Y\geq0\right\} ,\text{ }w:=\Sigma^{-1}\priorlocation.   
\end{equation}
For each vector $\beta\in\mathbb{R}^n$, let $\left\Vert \beta\right\Vert_{\Sigma}:=\sqrt{\beta^{\top}\Sigma\beta}$. Thus, $\wnorm=\sqrt{w^{\top}\Sigma w}=\sqrt{\priorlocation^{\top}\Sigma^{-1}\priorlocation}$. Denote by $\Phi(\cdot)$ the standard normal c.d.f. and by $\Phi^{-1}(\cdot)$ its inverse, i.e. the corresponding quantile function. 

\subsection{Ex-ante Robust Bayes Optimality}

\begin{thm}
[$\Gamma$-MMR optimal decisions]\label{thm:1}
Consider a treatment choice problem with welfare function \eqref{eq:welfare}, statistical model \eqref{eq:normal_model}, and set of priors \eqref{eq:Gamma}, that satisfy Assumption \ref{asm:1}.
\begin{itemize}
\item[(i)] If
\begin{equation}\label{eq:condition-1}
\frac{\overline{I}(\priorlocation)}{\overline{I}(\priorlocation)-\underline{I}(\priorlocation)}\geq\Phi(\wnorm),
\end{equation}

then $d_{w,0}^{*}$  is uniquely $\Gamma$-MMR optimal.
\item[(ii)] If
\begin{equation}
\frac{\overline{I}(\priorlocation)}{\overline{I}(\priorlocation)-\underline{I}(\priorlocation)}<\Phi(\wnorm),\label{eq:condition-large}
\end{equation}
then a rule $d\in\mathcal{D}_{n}$
attains $\Gamma$-MMR if, and only if, it implies 
\begin{eqnarray}
\mathbb{E}_{-\priorlocation}[d^{*}(Y)] & = & \frac{-\underline{I}(\priorlocation)}{\overline{I}(\priorlocation)-\underline{I}(\priorlocation)},\label{eq:mmr.large.condition}\\
\mathbb{E}_{\priorlocation}[d^{*}(Y)] & = & \frac{\overline{I}(\priorlocation)}{\overline{I}(\priorlocation)-\underline{I}(\priorlocation)}.\label{eq:mmr.large.condition.cont}
\end{eqnarray}
In particular, any rule of the form $\mathbf{1}\left\{ w^{\top}Y\geq c\right\}$ for some $c\in\mathbb{R}$ is \emph{not} $\Gamma$-MMR optimal. The following rules, and any convex combination of them, are $\Gamma$-MMR optimal:
\begin{align}
d^*_\text{RT}&:=\Phi\left(\frac{w^{\top}Y}{\tilde{\sigma}}\right),\\
 d^*_\text{linear}&:=\begin{cases}
0, & w^{\top}Y<-\rho^{*},\\
\frac{w^{\top}Y+\rho^{*}}{2\rho^{*}}, & -\rho^{*}\leq w^{\top}Y\leq\rho^{*},\\
1, & w^{\top}Y>\rho^{*},
\end{cases}\\
d^*_\text{step}&:=\begin{cases}
\frac{1}{2} - \beta^*, & w^{\top}Y<0,\\
\frac{1}{2} + \beta^*, & w^{\top}Y\geq0,\\
\end{cases}
\end{align}
where
\[\tilde{\sigma}=\sqrt{\left[\frac{\wnorm ^{2}}{\Phi^{-1}\left(\frac{\overline{I}(\priorlocation)}{\overline{I}(\priorlocation)-\underline{I}(\priorlocation)}\right)}\right]^{2}-\wnorm ^{2}},\]
$\rho^{*}>0$ is unique and solves $\int_{0}^{1}\Phi\left(\frac{2\rho^{*}x-\rho^{*}-\wnorm^{2}}{\wnorm }\right)dx=\frac{-\underline{I}(\priorlocation)}{\overline{I}(\priorlocation)-\underline{I}(\priorlocation)}$, and 
\[
\beta^{*}=\frac{\frac{\overline{I}(\bar{\mu})}{\overline{I}(\bar{\mu})-\underline{I}(\bar{\mu})}-\frac{1}{2}}{2\Phi\left(\wnorm \right)-1}\in\left(0,\frac{1}{2}\right).
\]

\item[(iii)] In case \eqref{eq:condition-large}, among linear-index threshold rules of the form
$\mathbf{1}\left\{ w^{\top}Y\geq c\right\}$, the optimal thresholds are  $\pm c^{*}$, where
\[
c^{*}:=\wnorm^2-\wnorm\Phi^{-1}\left(\frac{\overline{I}(\priorlocation)}{\overline{I}(\priorlocation)-\underline{I}(\priorlocation)}\right)>0.
\]
\item[(iv)] In case \eqref{eq:condition-large} and if $n>1$, the following linear-index threshold rule is also $\Gamma$-MMR
optimal:
\[
d_{w_{t^*},0}^{*}(Y):=\mathbf{1}\left\{ w_{t^*}^{\top}Y\geq0\right\} ,
\]
where $w_{t^*}=\Sigma^{-1}\left(t^*\priorlocation+(1-t^*)\dot{\mu}\right)$,
 $\dot{\mu}\neq0$ is such that $\dot{\mu}^{\top}\Sigma^{-1}\priorlocation=0$,
\begin{equation}\label{eq:t.star}
t^*:=\frac{1}{1\pm\sqrt{\frac{(1-s^{*})}{s^{*}}\frac{\wnorm^{2}}{\wdotnorm ^{2}}}},
\end{equation}
and 
\[
s^{*}:=\frac{\left[\Phi^{-1}\left(\frac{\overline{I}(\priorlocation)}{\overline{I}(\priorlocation)-\underline{I}(\priorlocation)}\right)\right]^{2}}{\wnorm^{2}}\in(0,1).
\]
Thus, $d_{w_{t^*},0}^{*}$ is also optimal among all linear threshold rules.
\end{itemize}
\end{thm}

Theorem \ref{thm:1} reveals that the $\Gamma$-MMR rules qualitatively change depending on whether condition (\ref{eq:condition-1}) is met or not. This condition admits an intuitive interpretation: Up to clamping to the unit interval (i.e., values outside this interval are mapped to its edges), the left-hand side of (\ref{eq:condition-1}) equals the unique minimax regret optimal rule for \emph{known} $\mu$ \citep{Manski2007}; therefore, it arguably measures the model's identification strength. The right-hand side of (\ref{eq:condition-1}) can be interpreted as the informativeness of the data about which of $\{-\priorlocation,\priorlocation\}$ obtained. Therefore, Theorem \ref{thm:1} says that if the model's identification power is sufficiently large compared to the informativeness of the data, then, the unique $\Gamma$-MMR optimal rule is $d_{w,0}^{*}$, a non-randomized linear index rule with threshold $0$ that effectively ignores the partial identification issue. In contrast, if the model's identification power is small compared to the informativeness of data, there are infinitely many $\Gamma$-MMR optimal rules, all of which satisfy (\ref{eq:mmr.large.condition}) and (\ref{eq:mmr.large.condition.cont}). Examples include suitably smoothed versions of $d_{w,0}^{*}$ like $d_{\text{RT}}^{*}$ and $d_{\text{linear}}^{*}$ (the functional forms of which showed up in \citealt{MQS}) as well as $d_{\text{step}}^{*}$  (which is a new result).

To understand how these two ``regimes'' arise, it is helpful to think of a zero-sum game against an adversarial Nature in which a MMR decision rule and a distribution over parameter values (the \textit{least favorable prior}) form a Nash equilibrium.\footnote{Analyzing this game is also how results are formally proved. See \cite{wald45} for what may be the first clear statement of this and \citet[][Theorem 1.4, Chap. 5]{lehmann2006theory} for a formalization. Within the literature on decisions under partial identification, the proof technique was first explicitly used in \cite{Stoye07} and the equilibrium structure with a noninformative and an informative regime was first encountered in \citet{stoye2012minimax}.} If identification is strong in the sense of \eqref{eq:condition-1}, this game has an \textit{informative} equilibrium: the least favorable prior evenly randomizes over two points $(\mu,U(\theta)) \in \{(\bar{\mu},\overline{I}(\bar{\mu})),(-\bar{\mu},\underline{I}(-\bar{\mu}))\}$, to which the unique Bayes response (and, therefore, uniquely optimal MMR rule) is  $d_{w,0}^*$. In all other cases, the equilibrium is \textit{uninformative}: the least favorable prior is supported on four points $(\mu, U(\theta))\in\{(\bar{\mu},\overline{I}(\bar{\mu})),(\bar{\mu},\underline{I}(\bar{\mu})),(-\bar{\mu},\overline{I}(-\bar{\mu})),(-\bar{\mu},\underline{I}(-\bar{\mu}))\}$ with a probability profile such that the posterior expectation of $U(\theta)$ always equals $0$. While any decision rule best responds to this prior, not every decision rule is MMR because the least favorable prior must also best respond to the decision rule. In this very structured setting, the latter is guaranteed by conditions \eqref{eq:mmr.large.condition} and \eqref{eq:mmr.large.condition.cont}, establishing claim (ii) and nonuniqueness of optimal decision rules. Indeed, even the nonrandomized threshold rules from Theorem \ref{thm:1}(iv) do not reflect any updating along the game's equilibrium path; they rather use uninformative features of the data as randomization device. 

In the uninformative equilibrium, we encounter several additional findings. First, despite the problem's apparent symmetry, $d_{w,0}^*$ is  \emph{not} optimal even among linear-index threshold rules that use the index $w^{\top}Y$. Instead, Theorem \ref{thm:1}(iii) characterizes exactly two optimal thresholds, one positive and one negative. Second, while the particular linear-index threshold rule  $d_{w,0}^*$  is not $\Gamma$-MMR optimal, in higher dimensional problems ($n>1$), there do exist linear-index threshold rules that are. This finding is in stark contrast to \citet{MQS}, who find that, in a large class of special cases, \emph{no} linear index threshold rule is MMR optimal. The crucial difference in settings is that the set of priors much constrains the decision theoretic problem's state space; as a result, the signal space is much richer than the state space, and this can be exploited to mimic randomization without nominally randomizing. Compare \citeauthor{manski2022identification}'s (\citeyear{manski2022identification}) abstract observation that, if a policy maker is not allowed to explicitly randomize, sampling uncertainty can be beneficial by providing an implicit randomization device. We note that this phenomenon is reminiscient of classic ``purification'' results in game theory \citep{DWW,purify}.\footnote{We thank Elliot Lipnowski for reminding us of this literature.} In contrast, it is not deeply related to the general intuition that ``Bayesians don't randomize.''

We next apply Theorem \ref{thm:1} to Example \ref{ex:stoye} and immediately get the following results.
\begin{corollary}[$\Gamma$-MMR decisions in Example \ref{ex:stoye}]\label{coro:1}
In Example \ref{ex:stoye}, the following statements are true:
    \begin{itemize}
        \item[(i)] If
        \begin{equation}\label{eq:condition}
  \frac{\priorlocation+k}{2k}\geq\Phi\left(\priorlocation/\sigma\right),
\end{equation}
then 
\[
d_0^*(\cdot):=\mathbf{1}\{\hat{\mu}\geq0\}
\]is the unique $\Gamma$-MMR optimal decision rule. \item[(ii)] If \eqref{eq:condition} fails, then a rule $d \in \mathcal{D}_1$ attains $\Gamma$-MMR if, and only if, it implies
\begin{eqnarray}
    \mathbb{E}[d(\hat{\mu}) \mid \mu=-\priorlocation] &=& \frac{-\priorlocation+k}{2k},\label{eq:stoye.large.condition.1} \\
    \mathbb{E}[d(\hat{\mu}) \mid \mu=\priorlocation] &=& \frac{\priorlocation+k}{2k}.\label{eq:stoye.large.condition.2}
\end{eqnarray}
In particular, no linear threshold rule is $\Gamma$-MMR optimal. The following rules, and any convex combination of them,  are all  $\Gamma$-MMR optimal:
\begin{align}
d^*_\text{RT}&:=\Phi\left(\frac{\hat{\mu}}{\tilde{\sigma}}\right),\\
 d^*_\text{linear}&:=\begin{cases}
0, &\hat{\mu}<-\frac{\sigma^{2}\rho^{*}}{\bar{\mu}},\\
\frac{\bar{\mu}\hat{\mu}+\sigma^{2}\rho^{*}}{2\sigma^{2}\rho^{*}}, & -\frac{\sigma^{2}\rho^{*}}{\bar{\mu}}\leq \hat{\mu}\leq\frac{\sigma^{2}\rho^{*}}{\bar{\mu}},\\
1, & \hat{\mu}>\frac{\sigma^{2}\rho^{*}}{\bar{\mu}},
\end{cases}\\
d^*_\text{step}&:=\begin{cases}
\frac{1}{2} - \frac{\priorlocation}{2k\left(2\Phi\left(\frac{\priorlocation}{\sigma}\right)-1\right)}, & \hat{\mu}<0,\\
 \frac{1}{2} + \frac{\priorlocation}{2k\left(2\Phi\left(\frac{\priorlocation}{\sigma}\right)-1\right)}, & \hat{\mu}\geq0,\\
\end{cases}
\end{align}
where 
\[\tilde{\sigma}=\sigma\sqrt{\left[\frac{\priorlocation}{\sigma \Phi^{-1}\left(\frac{\priorlocation+k}{2k}\right)}\right]^{2}-1 },\]
and $\rho^{*}>0$ is unique and solves $\int_{0}^{1}\Phi\left(\frac{2\rho^{*}x-\rho^{*}-(\nicefrac{\priorlocation}{\sigma}) ^{2}}{\nicefrac{\priorlocation}{\sigma} }\right)dx=\frac{-\priorlocation+k}{2k}$.

\item[(iii)] In case (ii), the best linear threshold rules in terms of $\Gamma$-minimax regret  are $\mathbf{1}\{\hat{\mu}\geq \pm c^*\}$, where
\begin{equation*}
     c^* = \priorlocation - \sigma \Phi^{-1}\left(\frac{\priorlocation+k}{2k}\right).
 \end{equation*}
    \end{itemize}
\end{corollary}

\begin{figure}[http]
{\centering
\includegraphics[width=0.9\linewidth]{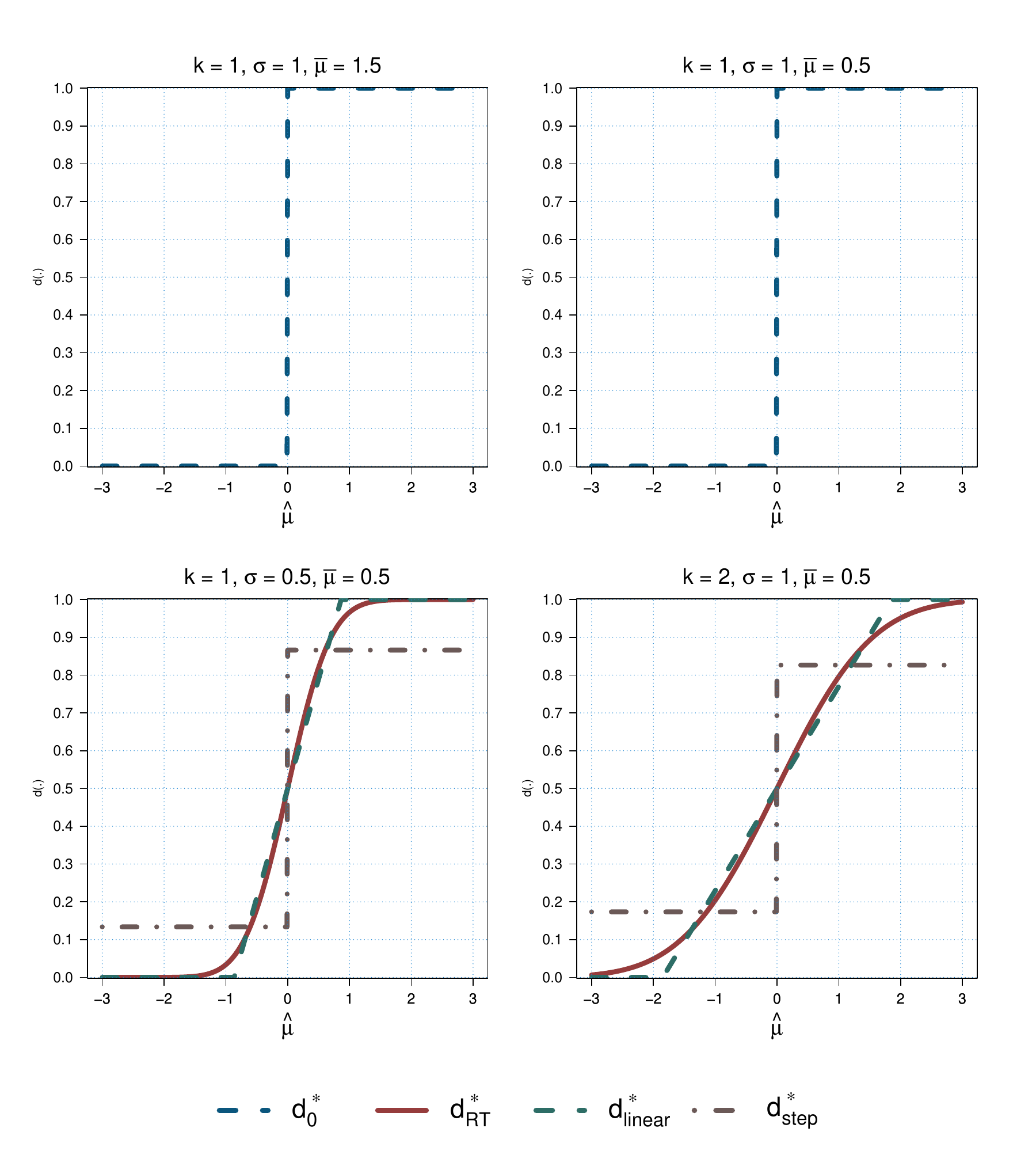}
    \caption{$\Gamma$-MMR optimal rules in Example \ref{ex:stoye}}
    \label{fig:ex-ante-example}}

    {\raggedright \footnotesize \textit{Notes}: This figure reports the $\Gamma$-MMR optimal rules in Example \ref{ex:stoye} for various combinations of parameter values. In the top two panels, the combinations of parameter values satisfy \eqref{eq:condition}. Therefore, the unique $\Gamma$-MMR optimal rule is $d^*_0$. In the bottom two panels, \eqref{eq:condition} fails. As a result,  $d^*_0$ is no longer $\Gamma$-MMR optimal. Instead, $d^*_{\text{RT}}$, $d^*_{\text{linear}}$ and $d^*_{\text{step}}$ are all $\Gamma$-MMR optimal. }
\end{figure}
Note that there is no analog to Theorem \ref{thm:1}'s case (iv); indeed, no linear threshold rule is optimal in case (ii). This is because the scalar nature of the signal $Y$ shuts down the aforementioned purification mechanism. See Figure \ref{fig:ex-ante-example} for an illustration of different $\Gamma$-MMR optimal rules in Example \ref{ex:stoye} for selected parameter values of $k,\sigma$ and $\bar{\mu}$.

\subsection{Ex-post Robust Bayes Optimality}\label{sec:ex.post}

\begin{thm}[$\Gamma$-PER optimal decisions]\label{thm:2}
Suppose all conditions in Theorem \ref{thm:1} hold true. Then: 
\begin{itemize}
\item[(i)] If $\underline{I}(\priorlocation)<0<\overline{I}(\priorlocation)$, the unique $\Gamma$-PER optimal rule is 
\[
d_{\operatorname{PER}}^{*}(Y)=\begin{cases}
\frac{-\underline{I}(\priorlocation)}{\overline{I}(\priorlocation)-\underline{I}(\priorlocation)}, & \text{if }w^{\top}Y<0,
\\
\frac{\overline{I}(\priorlocation)}{\overline{I}(\priorlocation)-\underline{I}(\priorlocation)}, & \text{if }w^{\top}Y\geq0.
\end{cases}
\]
Otherwise, 
$d_{w,0}^{*}$
is $\Gamma$-PER optimal.
\item[(ii)]   $d_{w,0}^{*}$ is always the
$\Gamma$-PER optimal non-randomized rule. 
\end{itemize}
\end{thm}

The results of Theorem \ref{thm:2} offer some important clarifications regarding $\Gamma$-PER in treatment choice problems with partial identification. First, even if regret is evaluated according to the posterior distribution,  it is not necessarily true that optimal rules are non-randomized. In fact, whenever there is model ambiguity regarding the sign of $U(\theta)$ (i.e., $\underline{I}(\priorlocation)<0<\overline{I}(\priorlocation)$), the unique $\Gamma$-PER optimal rule is randomized. Therefore, restricting the action space to $\{0,1\}$ in such problem is \emph{not} without loss of generality even under the $\Gamma$-PER criterion. Comparing results with Theorem \ref{thm:1} also allows for instructive observations on when $\Gamma$-PER and $\Gamma$-MMR optimal rules agree or disagree; we elaborate these in  Corollary \ref{coro:3}. Applying Theorem \ref{thm:2} to Example \ref{ex:stoye}, we finally obtain:

\begin{corollary}[$\Gamma$-PER rules in Example \ref{ex:stoye}] \label{coro:2}
In Example \ref{ex:stoye}, the following statements are true:
\begin{itemize}
        \item[(i)] If $\priorlocation<k$, then $\Gamma$-PER is uniquely minimized by
         \[d^*_{\operatorname{PER}}(\hat{\mu}) = \begin{cases}
        \frac{k + \bar{\mu}}{2k} & \text{if } \hat{\mu} \geq 0\\
        \frac{k-\bar{\mu}}{2k} & \text{if } \hat{\mu} < 0
\end{cases}.\]
\item[(ii)] If $\priorlocation\geq k$, then $d_0^*$ is $\Gamma$-PER optimal. Furthermore, it is always the $\Gamma$-PER optimal threshold rule.
\end{itemize}
\end{corollary}

\begin{figure}  
{\centering
\includegraphics[width=0.9\linewidth]{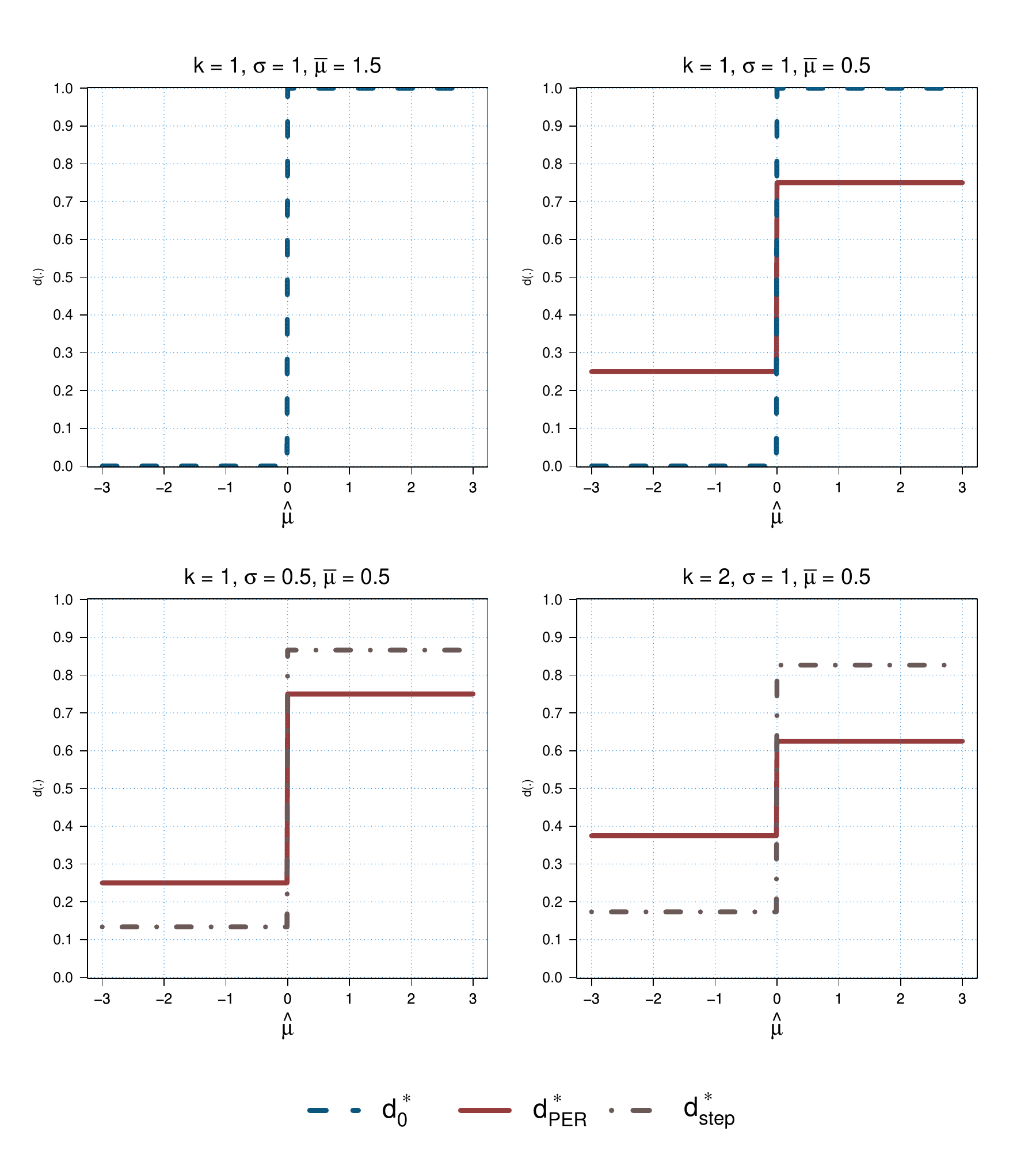}
    \caption{$\Gamma$-PER optimal rules in Example \ref{ex:stoye}}
    \label{fig:ex-post-example}}
    
 {\raggedright \footnotesize \textit{Notes}: This figure reports the $\Gamma$-PER optimal rules for the same parameter values considered in Figure \ref{fig:ex-ante-example}. In the top left panel, $\Gamma$-PER and $\Gamma$-MMR optimal rules coincide and are both $d^*_{0}$. For the rest of the panels, the unique $\Gamma$-PER optimal rules are all  $d^*_{\text{PER}}$ and are different from any $\Gamma$-MMR optimal rules.
}

\end{figure}

In Figure \ref{fig:ex-post-example}, we depict $\Gamma$-PER optimal rules for Example \ref{ex:stoye} with the same parameter values considered in Figure \ref{fig:ex-ante-example}. We see clearly that  $\Gamma$-MMR and -PER optimal rules coincide (if randomization is allowed) only in the special case when $k\leq\bar{\mu}$, an observation we generalize in Corollary \ref{coro:3}. More specifically, in the top left panel of Figure \ref{fig:ex-post-example}, as $k\leq\bar{\mu}$, $\Gamma$-MMR and -PER coincide and are the non-randomized threshold rule $d^*_{0}$. In the top right panel, $k>\bar{\mu}$ and the $\Gamma$-PER optimal rule becomes $d^*_{\text{PER}}$. However, since \eqref{eq:condition} still holds, $d^*_0$ is still $\Gamma$-MMR optimal. For the bottom two panels, as it still holds $k>\bar{\mu}$, $d^*_{\text{PER}}$ is still $\Gamma$-PER optimal. However, the associated parameter values imply \eqref{eq:condition} fails. As a result, $d^*_{0}$ is no longer $\Gamma$-MMR optimal and many  $\Gamma$-MMR rules exist. But even in theses cases, $\Gamma$-MMR and -PER rules differ, as among the class of step function rules (which contain $d^*_{\text{PER}}$), only $d^*_{\text{step}}$ is $\Gamma$-MMR optimal, still different from $d^*_{\text{PER}}$.

\begin{figure}[http]
\includegraphics[width=1\linewidth]{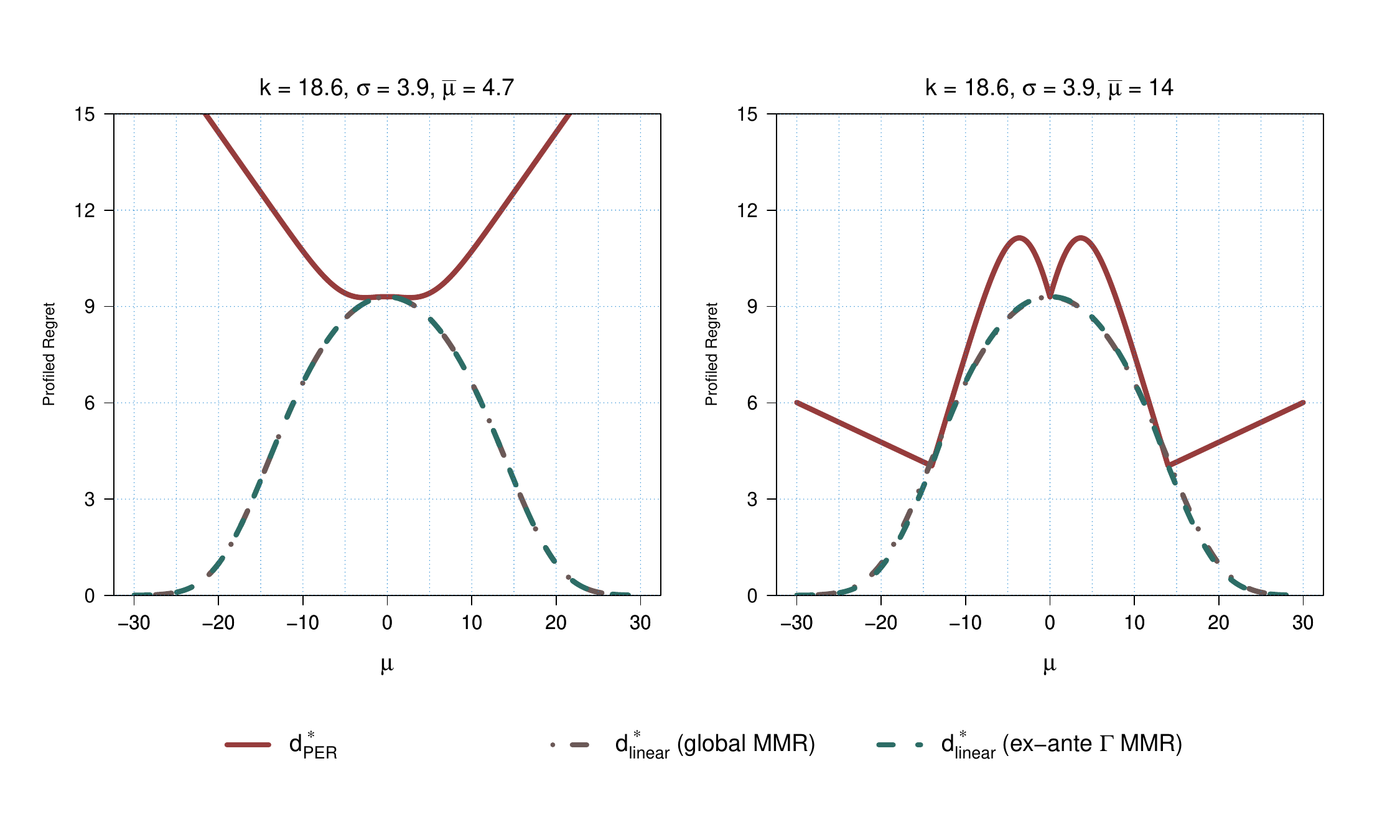}
\caption{Profiled regret of $\Gamma$-PER and other rules in Example \ref{ex:stoye}}\label{fig:profiled-regret-example}
 {\raggedright \footnotesize \textit{Notes}: This figure reports the (frequentist) profiled regrets, as a function of the true mean $\mu$ of data $\hat{\mu}$, of $\Gamma$-PER rule, $\Gamma$-MMR linear rule, and the least randomizing global MMR optimal rule (\citealt{MQS}) for the parameter values considered in \citet[Figure 2, ][]{MQS}. In the left plot,  we see $\Gamma$-PER rule is visually dominated. In both plots, the profiled regret of  $\Gamma$-MMR linear rule and the least randomizing global MMR optimal rule look very similar and are essentially overlapping with each other. 
}
\end{figure}

Applying \citet[Theorem 1,][]{MQS} to the current setting, we can conclude that  all rules, including both $\Gamma$-MMR and $\Gamma$-PER optimal rules, are at least admissible. Moreover, by definition, the $\Gamma$-MMR rule will have the lower value of the (ex ante) game.
  Therefore, it might be more useful to compare the (frequentist) profiled regrets \citep{MQS} of ex-post $\Gamma$-PER and other rules as a function of the true but unknown mean $\mu$ of data $\hat{\mu}$. In the context of Example \ref{ex:stoye}, we can report $\bar{R}(d,\mu):=\sup_{\mu^*\in[\mu-k,\mu+k]}R(d,\mu,\mu^*)$ as $\mu$ varies for each rule $d$, where $R(d,\mu,\mu^*)$ is the expected regret of rule $d$ as a function of $\mu$ and $\mu^{*}$. 
We report several findings (we emphasize that these are not obvious: we took an ex-ante perspective but did not
restrict ourselves to those values of $\mu$ that the prior allows). First, in both panels of Figure \ref{fig:profiled-regret-example}, the profiled expected regret of the $\Gamma$-PER rule exceeds that of the $\Gamma$-MMR rule. Therefore, in this particular example, $\Gamma$-MMR arguably outperforms $\Gamma$-PER from a broader frequentist point of view. Second, we in fact prove that,  whenever $\bar{\mu}$ is sufficiently small and $k$ is sufficiently large, the $\Gamma$-PER rule is  profiled-regret dominated, i.e., there exists a rule $d\neq d^*_{\text{PER}}$  such that 
\begin{equation}\label{eq:profiled.regret.dominance.main.text}
\overline{R}(d,\mu)\leq\overline{R}(d^*_{\text{PER}},\mu)
\end{equation}
for all $\mu\in\mathbb{R}$ with the inequality strict for some $\mu$;  see Lemma \ref{lem:profiled.regret} in Appendix \ref{sec:addtional.results}  for an exact statement. Intuitively, the $\Gamma$-PER  rule mixes between a coin flip rule and the naive threshold rule ($d_0^{*}$). When $\bar{\mu}$ is small, $\Gamma$-PER rule is more analogous to the coin flip rule, which \citet{MQS} show to be dominated in terms of profiled regret; also, its profiled regret fails to vanish as $\mu\to\pm\infty$ even though the optimal treatment is known ex ante in this case. 
Furthermore, Figure \ref{fig:profiled-regret-example} reveals that the $\Gamma$-MMR rule $d^{*}_{\text{linear}}$, though not globally MMR optimal, in some cases has essentially the same profiled risk function as the least randomizing global MMR optimal rule \citep{MQS}. This may lend a Robust Bayes interpretation to the latter.

\subsection{Ex-ante and Ex-post in the Game Against Nature}\label{sec:timing}

The different approaches analyzed here can all be expressed as different timing assumptions in the statistical game. The distinction is quite obvious for ex-ante versus ex-post: In the former (and original Waldian) perspective, this game is simultaneous move; in particular, Nature moves before data $Y$ are realized. In the ex-post perspective, Nature sees $Y$ before choosing a prior. It is immediately clear that this may be easier to solve because it allows for backward induction. There is also an immediate sense that solutions might not agree, as we indeed found.

But whether the decision maker is allowed to (or at least wants to) randomize or not can equally be thought of as changing the game's timing, providing another way to think about the action space for the ex-ante and ex-post approaches. The more standard perspective is that the decision maker may randomize over decision rules and Nature must move before learning the outcome of this randomization. By the nature of zero-sum games, this setup will frequently yield randomized solutions. In consequence, it is essential to define the  decision maker's action space as $[0,1]$. In contrast, if Nature is allowed to move after the decision maker's randomization is realized, then any incentive to randomize is gone and we may as well restrict the action space to $\{0,1\}$. %
 
Theorems \ref{thm:1} and \ref{thm:2} clarify that these distinctions actually matter in an interesting example. We spell this out in Corollary \ref{coro:3}, which considers both cases  when randomization is allowed and not allowed. The bottom line is that the assessment is quite sensitive to how the problem is set up. If underlying parameters lead to sufficiently small identification power, the criteria disagree. 

\begin{corollary}\label{coro:3}
Consider a treatment choice problem with welfare function \eqref{eq:welfare} and statistical model \eqref{eq:normal_model}  that satisfies Assumption \ref{asm:1}.
\begin{itemize}
    \item[(i)]  Suppose randomization is allowed. Then the $\Gamma$-MMR and $\Gamma$-PER optimal rules coincide if, and only if, $\underline{I}(\priorlocation)\geq0$.
    \item[(ii)] Suppose randomization is \emph{not} allowed. Then the $\Gamma$-MMR and $\Gamma$-PER optimal rules coincide if, and only if,
    \begin{equation}\label{eq:discussion.1}
\frac{\overline{I}(\priorlocation)}{\overline{I}(\priorlocation)-\underline{I}(\priorlocation)}\geq \Phi(\wnorm).
\end{equation}
\end{itemize}
\end{corollary}
Thus, if randomization is allowed, $\Gamma$-MMR and $\Gamma$-PER optimal rules coincide only in the somewhat trivial case in which we a priori know that $U(\theta)$ and $\mu$ have the same sign, so that optimal treatment choice reduces to Bayesian inference on the point identified $\mu$. They disagree in all other cases, including in settings where there are infinitely many $\Gamma$-MMR optimal rules.

One might expect more agreement once randomization is excluded; after all, this leads to a much simpler action space. Part (ii) shows that there is some truth to this: The condition for agreement changes from $\underline{I}(\priorlocation)\geq0$ to the strictly weaker \eqref{eq:discussion.1}. However, the criteria continue to disagree in many cases. It may be instructive to think of these cases in terms of ``comparative statics.'' For example, consider holding all parameters of the problem fixed but scaling the signal variance $\Sigma$ by a positive scalar, say to reflect a change in sample size. Then \eqref{eq:discussion.1} will hold if, and only if, $\Sigma$ is \emph{large} enough; hence, as long as $\underline{I}(\priorlocation)\geq0$, increasing sample size will eventually cause disagreement between $\Gamma$-MMR and -PER even if randomization is excluded. Similarly, for fixed $\Sigma$, $\Gamma$-MMR and -PER rules will always disagree if the model's identification power is \emph{small} enough.

\section{Conclusion}\label{sec:conclusion}
We studied treatment choice problems that display partial identification through the lens of the robust Bayes criteria. To do so, we take the general framework of \citet{yata2021} and others and embed in it a simple example of the set of priors advocated by \cite{GiacominiKitagawa}. We describe and contrast (ex-ante) $\Gamma$-minimax regret and (ex-post) $\Gamma$-posterior expected regret and analytically derive optimal solutions with and without randomization.

Our results contain two key messages that we think are valuable to the literature. First, with partial identification and multiple priors, ex-ante and ex-post assessments do \emph{not} agree in general, whether or not randomized rules are allowed. This may at first seem expected due to dynamic inconsistency of multiple prior Bayes criteria, but was not obvious in view of the set of prior's specific structure. Second, randomization can be optimal in both ex-ante and ex-post problems---it is with loss of generality to exclude them even when regret is evaluated ex-post. The contrast between the results also illustrates a need to better understand the comparative advantages---whether from a theoretical or practical perspective---of using one criterion over the other.

An obvious limitation lies in our use of a convenient but restrictive prior on $\mu$. As we discovered in Section \ref{sec:ex.post}, the $\Gamma$-MMR rule also performs well across other values of $\mu$; however, this may be related to the fact that the globally least favorable prior in this setting has a two-point structure on $\mu$ as well and would obviously not generalize to arbitrary uses of restrictive priors. We provide some insight on more general priors in Appendix \ref{sec:generality}. In short, the tractability advantage of $\Gamma$-PER may become pronounced in such settings, although we hope that recent computational developments \citep{fernandez2024epsilon,guggenberger2025numerical} will attenuate this concern.

\bibliographystyle{ecta}
\bibliography{refs}

\newpage
\appendix

\section{Proofs of Main Results}

Proofs will frequently claim and verify equilibria of the fictitious game against Nature. Recall that, from basic facts about zero-sum games, if a decision rule \emph{uniquely} best responds to some least favorable prior, it must be the \emph{unique} equilibrium rule.

\subsection{Proof of Theorem \ref{thm:1}}

\subsubsection{Statement (i)}

The expected regret of decision rule $d$ is 
\[
R(d,\theta)=U(\theta)\left(\mathbf{1}\{U(\theta)\geq0\}-\mathbb{E}_{m(\theta)}[d(Y)]\right), \quad \theta\in\Theta.
\]
Recall that, by the definition of $\Gamma$, we have $\pi_{\mu}\sim\text{unif}(\{-\bar{\mu},\bar{\mu}\})$ and, given $\mu=\pm\priorlocation$, $\pi_{\theta\mid\mu}(U(\theta)\in I(\mu))=1$.
Therefore, the Bayes expected regret of $d$ under prior $\pi\in\Gamma$ equals
\begin{align*}
r(d,\pi) & =\frac{1}{2}\cdot\left[\int_{\tilde{\theta}\in\Theta}U(\tilde{\theta})\left(\mathbf{1}\{U(\tilde{\theta})\geq0\}-\mathbb{E}_{\bar{\mu}}[d]\right)d\pi_{\theta\mid\bar{\mu}}(\tilde{\theta})\right]\\
 & +\frac{1}{2}\cdot\left[\int_{\tilde{\theta}\in\Theta}U(\tilde{\theta})\left(\mathbf{1}\{U(\tilde{\theta})\geq0\}-\mathbb{E}_{-\bar{\mu}}[d]\right)d\pi_{\theta\mid-\bar{\mu}}(\tilde{\theta})\right],
\end{align*}
where $\mathbb{E}_{\bar{\mu}}[d]:=\mathbb{E}_{\bar{\mu}}[d(Y)]$,  $\mathbb{E}_{-\bar{\mu}}[d]:=\mathbb{E}_{-\bar{\mu}}[d(Y)]$. One can easily solve for
\begin{align*}
\sup_{\pi\in\Gamma}r(d,\pi) & =\frac{1}{2}\max\left\{ \overline{I}(\bar{\mu})(1-\mathbb{E}_{\bar{\mu}}[d]),-\underline{I}(\bar{\mu})\mathbb{E}_{\bar{\mu}}[d]\right\} \\
 & +\frac{1}{2}\max\left\{ \overline{I}(-\bar{\mu})(1-\mathbb{E}_{-\bar{\mu}}[d]),-\underline{I}(-\bar{\mu})\mathbb{E}_{-\bar{\mu}}[d]\right\}. 
\end{align*}
The least favorable prior $\pi^*$ equals 
\begin{equation}
\pi^{*}=\left\{ \pi^{*}_{\mu}\sim\text{unif}(\{-\bar{\mu},\bar{\mu}\}),\pi^{*}_{\theta\mid\bar{\mu}}(U(\theta)=\overline{I}(\bar{\mu}))=1,\pi^{*}_{\theta\mid-\bar{\mu}}(U(\theta)=\underline{I}(-\bar{\mu}))=1\right\}.\label{pf:thm1.1}
\end{equation}
Lemma \ref{lemma:a1-1} shows that the unique Bayes rule against $\pi^*$ is 
\[
d_{w,0}^{*}=\mathbf{1}\{w^{\top}Y\geq0\},\quad\text{where }w=\Sigma^{-1}\bar{\mu}.
\]
Lemma \ref{lemma:a1-2} establishes that 
$\sup_{\pi\in\Gamma}r(d_{w,0}^{*},\pi)=r(d_{w,0}^{*},\pi^{*})$ 
as long as \eqref{eq:condition-1} holds true. This establishes the claim.

\subsubsection{Statement (ii)}

\paragraph{Step 1}

We show that when \eqref{eq:condition-large} holds, any rule $d\in\mathcal{D}_{n}$
is $\Gamma$-MMR optimal if 
\begin{equation}
\mathbb{E}_{\bar{\mu}}[d]=\frac{\overline{I}(\bar{\mu})}{\overline{I}(\bar{\mu})-\underline{I}(\bar{\mu})},\quad\mathbb{E}_{-\bar{\mu}}[d]=\frac{-\underline{I}(\bar{\mu})}{\overline{I}(\bar{\mu})-\underline{I}(\bar{\mu})}.\label{pf:thm.1.2}
\end{equation}
The least favorable prior $\pi^{*}$ is such that
 $\mu\sim\text{unif}(\{-\bar{\mu},\bar{\mu}\})$,
and when $\mu=\bar{\mu}$, 
\[
U(\theta)=\begin{cases}
\overline{I}(\bar{\mu}), & \text{with probability (w.p.) }p_{1},\\
\underline{I}(\bar{\mu}), & \text{w.p. }1-p_{1},
\end{cases}
\]
where $p_{1}>0$ is such that $p_{1}\overline{I}(\bar{\mu})+(1-p_{1})\underline{I}(\bar{\mu})=0$,
and when $\mu=-\bar{\mu}$, 
\[
U(\theta)=\begin{cases}
\overline{I}(-\bar{\mu}), & \text{w.p. }p_{2},\\
\underline{I}(-\bar{\mu}), & \text{w.p. }1-p_{2},
\end{cases}
\]
where $p_{2}$ is such that $p_{2}\overline{I}(-\bar{\mu})+(1-p_{2})\underline{I}(-\bar{\mu})=0$.
Lemma \ref{lemma:a1-2-1} establishes that any decision rule is Bayes against this prior (intuitively because the data are uninformative), and Lemma \ref{lemma:a1-2-2}
further shows that, for any rule $d$ that satisfies \eqref{pf:thm.1.2}, $\sup_{\pi\in\Gamma}r(d,\pi)=r(d,\pi^*)$ obtains. This establishes the claim.

\paragraph{Step 2}

We next verify the ``only if'' statement. Recall that the least favorable prior $\pi^*$ must best respond to any optimal decision rule, i.e., for any MMR optimal
rule $d$, $\pi^*$ must solve 
\begin{align*}
\sup_{\pi\in\Gamma}r(d,\pi) & =\frac{1}{2}\max\left\{ \overline{I}(\bar{\mu})(1-\mathbb{E}_{\bar{\mu}}[d]),-\underline{I}(\bar{\mu})\mathbb{E}_{\bar{\mu}}[d]\right\} \\
 & +\frac{1}{2}\max\left\{ \overline{I}(-\bar{\mu})(1-\mathbb{E}_{-\bar{\mu}}[d]),-\underline{I}(-\bar{\mu})\mathbb{E}_{-\bar{\mu}}[d]\right\}.
\end{align*}
This, however, requires that Nature is indifferent between $\overline{I}(\bar{\mu})$ and $\underline{I}(\bar{\mu})$ when $\mu=\bar{\mu}$ and similarly between $\overline{I}(-\bar{\mu})$ and $\underline{I}(-\bar{\mu})$ when $\mu=-\bar{\mu}$.
That is, we must have 
\[
\overline{I}(\bar{\mu})(1-\mathbb{E}_{\bar{\mu}}[d])=-\underline{I}(\bar{\mu})\mathbb{E}_{\bar{\mu}}[d],\quad\overline{I}(-\bar{\mu})(1-\mathbb{E}_{-\bar{\mu}}[d])=-\underline{I}(-\bar{\mu})\mathbb{E}_{-\bar{\mu}}[d],
\]
which is equivalent to \eqref{pf:thm.1.2}. 

\paragraph{Step 3}

We show a rule of form $\mathbf{1}\left\{ w^{\top}Y\geq c\right\} $
for some $c\in\mathbb{R}$ cannot be $\Gamma$-MMR optimal when \eqref{eq:condition-large}
holds. Note $w=\Sigma^{-1}\bar{\mu}$. Thus, 
\[
\mathbb{E}_{\mu}[\mathbf{1}\{w^{\top}Y\geq c\}]=1-\Phi\left(\frac{c-w^{\top}\mu}{\sqrt{w^\top \Sigma w}}\right).
\]
Suppose by contradiction that a rule $\mathbf{1}\left\{ w^{\top}Y\geq c\right\} $ is optimal, then by statement (ii) we have
\begin{align}
\qquad & 1-\Phi\left(\frac{c-w^\top \bar{\mu}}{\sqrt{w^\top \Sigma w}}\right)=\frac{\overline{I}(\bar{\mu})}{\overline{I}(\bar{\mu})-\underline{I}(\bar{\mu})}\label{pf:thm1.3.1}\\
\qquad & 1-\Phi\left(\frac{c+w^\top \bar{\mu}}{\sqrt{w^\top \Sigma w}}\right)=\frac{-\underline{I}(\bar{\mu})}{\overline{I}(\bar{\mu})-\underline{I}(\bar{\mu})}.\label{pf:thm1.3.2}
\end{align}
By symmetry, \eqref{pf:thm1.3.1} and \eqref{pf:thm1.3.2} can both hold only if $c=0$. But \eqref{eq:condition-large} then implies that
\[
\Phi\left(\frac{w^\top \bar{\mu}}{\sqrt{w^\top \Sigma w}}\right)=\Phi(\wnorm)>\frac{\overline{I}(\bar{\mu})}{\overline{I}(\bar{\mu})-\underline{I}(\bar{\mu})},
\]
so that \eqref{pf:thm1.3.1} cannot in fact hold when $c=0$, a contradiction.

\paragraph{Step 4}

We verify that $d_{\text{RT}}^*$, $d_{\text{linear}}^*$ and $d_{\text{step}}^*$ are all $\Gamma$-MMR optimal. Due to symmetry,
it suffices to show that 
\[
\mathbb{E}_{\priorlocation}[d_{\text{RT}}^{*}]=\mathbb{E}_{\priorlocation}[d_{\text{linear}}^{*}]=\mathbb{E}_{\priorlocation}[d_{\text{step}}^{*}]=\frac{\overline{I}(\bar{\mu})}{\overline{I}(\bar{\mu})-\underline{I}(\bar{\mu})}.
\]
To see $\mathbb{E}_{\priorlocation}[d_{\text{RT}}^*]=\frac{\overline{I}(\bar{\mu})}{\overline{I}(\bar{\mu})-\underline{I}(\bar{\mu})}$, consider the random threshold rule $\mathbf{1}\left\{ w^\top Y\geq\xi\right\}$, where $\xi\sim N(0,\tilde{\sigma}^{2})$ is independent of $Y$. As $w^\top Y-\xi\sim N\left(\wnorm ^2,\wnorm^2+\tilde{\sigma}^2\right)$, algebra shows 
\[
\mathbb{E}_{\priorlocation}\left[\mathbf{1}\left\{ w^\top Y\geq\xi\right\} \right]=\frac{\overline{I}(\bar{\mu})}{\overline{I}(\bar{\mu})-\underline{I}(\bar{\mu})}
\]
as required. To see $\mathbb{E}_{\priorlocation}[d_{\text{linear}}^*]=\frac{\overline{I}(\bar{\mu})}{\overline{I}(\bar{\mu})-\underline{I}(\bar{\mu})}$,
let $\rho>0$ and 
\[
d_{\text{linear},\rho}:=\begin{cases}
0, & w^{\top}Y<-\rho,\\
\frac{w^{\top}Y+\rho^{*}}{2\rho^{*}}, & -\rho\leq w^{\top}Y\leq\rho,\\
1, & w^{\top}Y>\rho.
\end{cases}
\]
Applying Lemma B.7 in \citet{MQS} and integration by parts yield
\[
f(\rho):=\mathbb{E}_{\priorlocation}[d_{\text{linear},\rho}]=  1-\int_{0}^{1}\Phi\left(\frac{2\rho x-\rho-\wnorm ^{2}}{\wnorm }\right)dx
=  1-\frac{\wnorm }{2\rho}\int_{\frac{-\rho-\wnorm ^{2}}{\wnorm }}^{\frac{\rho-\wnorm ^{2}}{\wnorm }}\Phi\left(t\right)dt.
\]
Note that $\lim_{\rho\downarrow0}f(\rho)=1-\Phi(-\wnorm) = \Phi(\wnorm)>\frac{\overline{I}(\bar{\mu})}{\overline{I}(\bar{\mu})-\underline{I}(\bar{\mu})}$,
while L'Hopital's rule implies 
\begin{align*}
\lim_{\rho\rightarrow\infty}f(\rho)  =1-\frac{1}{2}\lim_{\rho\rightarrow\infty}\left\{\Phi\left(\frac{\rho-\wnorm ^{2}}{\wnorm }\right)-\Phi\left(\frac{-\rho-\wnorm ^{2}}{\wnorm }\right)\right\}
  =\frac{1}{2}<\frac{\overline{I}(\bar{\mu})}{\overline{I}(\bar{\mu})-\underline{I}(\bar{\mu})},
\end{align*}
where the last last inequality follows from $\overline{I}(\bar{\mu})+\underline{I}(\bar{\mu})>0$
and $1>\Phi\left(\wnorm \right)>\frac{\overline{I}(\bar{\mu})}{\overline{I}(\bar{\mu})-\underline{I}(\bar{\mu})}$.
Furthermore, by applying the chain rule, $\frac{\partial f(\rho)}{\partial\rho}<0$.
Therefore, $f(\cdotp)$ is strictly decreasing in $(0,\infty)$. We
conclude that there must exist some unique $\rho^{*}>0$ such that
$f(\rho^{*})=\frac{\overline{I}(\bar{\mu})}{\overline{I}(\bar{\mu})-\underline{I}(\bar{\mu})}$,
or equivalently, $\int_{0}^{1}\Phi\left(\frac{2\rho^{*}x-\rho^{*}-\wnorm ^{2}}{\wnorm }\right)dx=\frac{-\overline{I}(\bar{\mu})}{\overline{I}(\bar{\mu})-\underline{I}(\bar{\mu})}$,
which implies $\mathbb{E}_{\priorlocation}[d_{\text{linear}}^{*}]=\frac{\overline{I}(\bar{\mu})}{\overline{I}(\bar{\mu})-\underline{I}(\bar{\mu})}$.

Finally, we verify that $\mathbb{E}_{\priorlocation}[d_{\text{step}}^{*}]=\frac{\overline{I}(\bar{\mu})}{\overline{I}(\bar{\mu})-\underline{I}(\bar{\mu})}$.
For any $\beta\in\left(0,\frac{1}{2}\right)$, consider the following
step function rule: 
\[
d_{\text{step},\beta}:=\left(\frac{1}{2}-\beta\right)\mathbf{1}\left\{w^{\top}Y<0\right\}+\left(\frac{1}{2}+\beta\right)\mathbf{1}\left\{w^{\top}Y\geq0\right\}.
\]
One then has 
\begin{align*}
\mathbb{E}_{\priorlocation}[d_{\text{step},\beta}]= & \left(\frac{1}{2}-\beta\right)\Phi\left(-\wnorm \right)+\left(\frac{1}{2}+\beta\right)\left(1-\Phi\left(-\wnorm \right)\right)\\
= & \frac{1}{2}+\beta\left(2\Phi\left(\wnorm \right)-1\right).
\end{align*}
Setting $\mathbb{E}_{\priorlocation}[d_{\text{step},\beta^{*}}]=\frac{\overline{I}(\bar{\mu})}{\overline{I}(\bar{\mu})-\underline{I}(\bar{\mu})}$
yields $\beta^{*}=\frac{\frac{\overline{I}(\bar{\mu})}{\overline{I}(\bar{\mu})-\underline{I}(\bar{\mu})}-\frac{1}{2}}{2\Phi\left(\wnorm \right)-1}$.
As $\overline{I}(\bar{\mu})+\underline{I}(\bar{\mu})>0$
and (\ref{eq:condition-large}) holds, $\frac{\overline{I}(\bar{\mu})}{\overline{I}(\bar{\mu})-\underline{I}(\bar{\mu})}>\frac{1}{2}$
and therefore $\beta^{*}>0$. Furthermore, $\beta^{*}<\frac{1}{2}$
holds due to (\ref{eq:condition-large}) as well. Since $d_{\text{step}}^{*}=d_{\text{step},\beta^{*}}$,
we conclude that $\mathbb{E}_{\priorlocation}[d_{\text{step}}^{*}]=\frac{\overline{I}(\bar{\mu})}{\overline{I}(\bar{\mu})-\underline{I}(\bar{\mu})}$.

\subsubsection{Statement (iii)}

For any rule of form $d_{w,c}(Y):=\mathbf{1}\{w^{\top}Y\geq c\}$
where $w=\Sigma^{-1}\bar{\mu}$ and $c\in\mathbb{R}$, we may calculate
\[
\mathbb{E}_{\bar{\mu}}[d_{w,c}(Y)]=1-\Phi\left(\frac{c-w^{\top}\bar{\mu}}{\sqrt{w^{\top}\Sigma w}}\right)
\]
and (due to Lemma \ref{lemma:b1}, recalling $\wnorm ^{2}=w^{\top}\Sigma w=\bar{\mu}^{\top}\Sigma^{-1}\bar{\mu}$
) 
\begin{align*}
g(c):=\sup_{\pi\in\Gamma} r(d_{w,c},\pi) & =\frac{1}{2}\max\left\{ \overline{I}(\bar{\mu})\Phi\left(\frac{-\wnorm ^{2}+c}{\wnorm }\right),-\underline{I}(\bar{\mu})\Phi\left(\frac{\wnorm ^{2}-c}{\wnorm }\right)\right\} \\
 & +\frac{1}{2}\max\left\{ -\underline{I}(\bar{\mu})\Phi\left(\frac{\wnorm ^{2}+c}{\wnorm }\right),\overline{I}(\bar{\mu})\Phi\left(\frac{-\wnorm ^{2}-c}{\wnorm }\right)\right\} .
\end{align*}
Lemma \ref{lemma:a1-3-1} shows that $g$ is decreasing on $[0,c^*]$ and increasing on $[c^{*},\infty)$, implying that the optimal threshold rule is $d_{w,c^{*}}$ when $c\in[0,\infty)$. By symmetry, $d_{w,-c^*}$ is optimal when $c\in(-\infty,0]$, and $d_{w,-c^*}$ and $d_{w,c^*}$ share the same worst-case expected regret.

\subsubsection{Statement (iv)}

In case (\ref{eq:condition-large}) and when $n>1$, there exists $\dot{\mu}\neq\mathbf{0}$ is such that $\dot{\mu}^{\top}\Sigma^{-1}\bar{\mu}=0$,
i.e., $\dot{\mu}$ is orthogonal to $\Sigma^{-1}\bar{\mu}$. For any $t\in\mathbb{R}$, let 
\begin{equation}
d_{w_{t},0}(Y)=\mathbf{1}\{w_{t}^{\top}Y\geq0\},\text{ where }w_{t}=\Sigma^{-1}(t\bar{\mu}+(1-t)\dot{\mu}).\label{pf:thm1.4.1}
\end{equation}
Lemma \ref{lemma:a1-4-1} shows that when $t=t^{*}$, $\mathbb{E}_{\bar{\mu}}[d_{w_{t^{*}},0}(Y)]=\frac{\overline{I}(\bar{\mu})}{\overline{I}(\bar{\mu})-\underline{I}(\bar{\mu})}$.
By symmetry, one then also has $\mathbb{E}_{-\bar{\mu}}[d_{w_{t^{*}},c}(Y)]=\frac{-\underline{I}(\priorlocation)}{\overline{I}(\priorlocation)-\underline{I}(\priorlocation)}$.
Applying statement (ii) yields that $d_{w_{t^{*}},0}$ is MMR optimal.

\subsection{Proof of Corollary \ref{coro:1}}

In Example \ref{ex:stoye}, $\overline{I}(\priorlocation)=\priorlocation+k,\underline{I}(\priorlocation)=\priorlocation-k$,
$\Sigma=\sigma^{2}$. The results of the corollary follow directly from Theorem \ref{thm:1}(i)-(iii).

\subsection{Proof of Theorem \ref{thm:2}}

Recall that the $\Gamma$-PER optimal rule solves 
\[
\inf_{a\in[0,1]}\sup_{\pi\in\Gamma}\int_{\tilde{\theta}\in\Theta}L(a,\tilde{\theta})d\pi_{\theta \mid Y}(\tilde{\theta}),\hspace{1em}\forall Y\in\mathbb{R}^{n},
\]
where $L(a,\theta)=U(\theta)(\mathbf{1}\{U(\theta)\geq0\}-a)$, and $\pi_{\theta\mid Y}$ is the posterior distribution of $\theta$ given $Y$. If randomization is not allowed, the $\Gamma$-PER optimal rule solves
\[
\inf_{a\in\{0,1\}}\sup_{\pi\in\Gamma}\int_{\tilde{\theta}\in\Theta}L(a,\tilde{\theta})d\pi_{\theta \mid Y}(\tilde{\theta}),\hspace{1em}\forall Y\in\mathbb{R}^{n}.
\]
Statement (i) then follows from Lemma \ref{lemma:a2-1}; statement (ii) follows from Lemma \ref{lemma:a2-2}.

\subsection{Proof of Corollary \ref{coro:2}}

Directly follows from Theorem \ref{thm:2}.

\subsection{Proof of Corollary \ref{coro:3}}

\paragraph*{(i)}

``If'': If $\underline{I}(\bar{\mu})\geq0$, then Theorem \ref{thm:1}(i) and Theorem \ref{thm:2}(i) apply and establish that $d_{w,0}^*$ is both the unique $\Gamma$-MMR and the unique $\Gamma$-PER optimal rule.

``Only if'': If $\underline{I}(\bar{\mu})<0$, then $d_{\text{PER}}^{*}$ is uniquely $\Gamma$-PER optimal by Theorem \ref{thm:2}(i). If condition \eqref{eq:condition-1} holds as well, then Theorem \ref{thm:1}(i) implies that $d_{w,0}^*$ is uniquely $\Gamma$-MMR optimal; hence, $\Gamma$-MMR and $\Gamma$-PER optimal rules disagree. If condition \eqref{eq:condition-large} applies, then, by Theorem \ref{thm:1}(ii), a $\Gamma$-MMR optimal rule must satisfy \eqref{eq:mmr.large.condition}-\eqref{eq:mmr.large.condition.cont}. But $d_{\text{PER}}^*$ can be written as 
\[
d_{\operatorname{PER}}^{*}=d_{\operatorname{step},\beta_{\text{PER}}}=\begin{cases}
\frac{1}{2}-\beta_{\text{PER}}, & \text{if }w^{\top}Y<0,\\
\frac{1}{2}+\beta_{\text{PER}}, & \text{if }w^{\top}Y\geq0,
\end{cases}
\]
where $\beta_{\text{PER}}=\frac{\overline{I}(\priorlocation)}{\overline{I}(\priorlocation)-\underline{I}(\priorlocation)}-\frac{1}{2}$. Step 4 for the proof of Theorem \ref{thm:1}(ii) shows that the only $\Gamma$-MMR optimal rule of form $d_{\text{step},\beta}$ is $d_{\text{step}}^*\neq d_{\text{PER}}^*$.

\paragraph*{(ii)}

``If'': When randomization is not allowed, Theorem \ref{thm:2}(ii) shows that $d_{w,0}^*$ is always $\Gamma$-PER optimal. Under condition \eqref{eq:discussion.1}, Theorem \ref{thm:1}(i) establishes that $d_{w,0}^*$ is $\Gamma$-MMR optimal as well.

``Only if'': Again, $d_{w,0}^*$ is always $\Gamma$-PER optimal when randomization is not allowed. If \eqref{eq:discussion.1} fails (i.e., \eqref{eq:condition-large} holds) and $n>1$, it follows by Theorem \ref{thm:1}(iv) that the randomized rule $d_{w_{t^*},0}^*$ is $\Gamma$-MMR optimal. If $n=1$, Theorem \ref{thm:1}(iii) implies that $d_{w,0}^*$ is not optimal even among linear threshold rules. Therefore, when \eqref{eq:condition-large} holds, $\Gamma$-MMR and $\Gamma$-PER optimal rules disagree with and without randomization.

\section{Technical Lemmas Supporting Ex-ante Analysis}

\begin{lemma}\label{lemma:a1-1} The Bayes rule supported by prior \eqref{pf:thm1.1} is $d_{w,0}^*$. \end{lemma}
\begin{proof}
Note $\overline{I}(\bar{\mu})>0$ due to  $\overline{I}(\priorlocation)+\underline{I}(\priorlocation)>0$ and $\underline{I}(-\bar{\mu})=-\overline{I}(\bar{\mu})$ by Lemma \ref{lemma:b1}. Given $\pi^{*}$, the Bayes optimal rule must solve the posterior problem
\begin{eqnarray*}
    &\min_{a\in[0,1]}&\int_{\tilde{\theta}\in\Theta}L(a,\tilde{\theta})d\pi^*_{\theta \mid Y}(\tilde{\theta}), \\
    &&\int_{\tilde{\theta}\in\Theta}L(a,\tilde{\theta})d\pi^*_{\theta \mid Y}(\tilde{\theta})\propto\overline{I}(\bar{\mu})(1-a)\cdot\frac{1}{2}\cdot f(Y|\bar{\mu})+\underline{I}(-\bar{\mu})(-a)\cdot\frac{1}{2}\cdot f(Y|-\bar{\mu}),
\end{eqnarray*}
where $f(Y|\bar{\mu})$ and $f(Y|-\bar{\mu})$ are the likelihood of $Y$ at $\bar{\mu}$ and $-\bar{\mu}$. This problem is equivalent to
\begin{align*}
\min_{a\in[0,1]}\overline{I}(\bar{\mu})f(Y|\bar{\mu})+a\underbrace{\overline{I}(\bar{\mu})}_{>0}(f(Y|-\bar{\mu})-f(Y|\bar{\mu})).
\end{align*}
Since $\overline{I}(\bar{\mu})>0$,
the unique Bayes optimal rule is $\mathbf{1}\{f(Y|-\bar{\mu})-f(Y|\bar{\mu})\leq0\}$, which is equivalent to $d^*_{w,0}$ after further algebra. 
\end{proof}
\begin{lemma}\label{lemma:a1-2} Consider decision rule $d^*_{w,0}$
and prior $\pi^{*}$ defined in \eqref{pf:thm1.1}. If \eqref{eq:condition-1}
holds, then 
\begin{equation}
\sup_{\pi\in\Gamma}r(d_{w,0}^{*},\pi)=r(d_{w,0}^{*},\pi^{*}).\label{pf:lem.b2.1}
\end{equation}

\end{lemma} 
\begin{proof}
As $w^{\top}Y\sim N(w^{\top}\mu,w^{\top}\Sigma w)$, algebra shows
\begin{align*}
\mathbb{E}_{\mu}[d_{w,0}^{*}] & =\Phi\left(\frac{w^{\top}\mu}{\sqrt{w^{\top}\Sigma w}}\right)
\end{align*}
for all $\mu\in M$. In particular, 
\begin{align*}
\mathbb{E}_{\bar{\mu}}[d_{w,0}^{*}]=\Phi\left(\frac{w^{\top}\bar{\mu}}{\sqrt{w^{\top}\Sigma w}}\right),\quad\mathbb{E}_{-\bar{\mu}}[d_{w,0}^{*}]=\Phi\left(-\frac{w^{\top}\bar{\mu}}{\sqrt{w^{\top}\Sigma w}}\right).
\end{align*}
It follows that 
\begin{align*}
\sup_{\pi\in\Gamma}r(d_{w,0}^{*},\pi) & =\frac{1}{2}\max\left\{ \overline{I}(\bar{\mu})\Phi\left(-\frac{w^{\top}\bar{\mu}}{\sqrt{w^{\top}\Sigma w}}\right),-\underline{I}(\bar{\mu})\Phi\left(\frac{w^{\top}\bar{\mu}}{\sqrt{w^{\top}\Sigma w}}\right)\right\} \\
 & +\frac{1}{2}\max\left\{ \overline{I}(-\bar{\mu})\Phi\left(\frac{w^{\top}\bar{\mu}}{\sqrt{w^{\top}\Sigma w}}\right),-\underline{I}(-\bar{\mu})\Phi\left(-\frac{w^{\top}\bar{\mu}}{\sqrt{w^{\top}\Sigma w}}\right)\right\},
\end{align*}
and 
\[
r(d_{w,0}^{*},\pi^{*})=\frac{1}{2}\overline{I}(\bar{\mu})\Phi\left(-\frac{w^{\top}\bar{\mu}}{\sqrt{w^{\top}\Sigma w}}\right)-\frac{1}{2}\underline{I}(-\bar{\mu})\Phi\left(-\frac{w^{\top}\bar{\mu}}{\sqrt{w^{\top}\Sigma w}}\right).
\]
So, \eqref{pf:lem.b2.1} holds as long as 
\begin{itemize}
\item[(i)] $\overline{I}(\bar{\mu})\Phi\left(-\frac{w^{\top}\bar{\mu}}{\sqrt{w^{\top}\Sigma w}}\right)\geq-\underline{I}(\bar{\mu})\Phi\left(\frac{w^{\top}\bar{\mu}}{\sqrt{w^{\top}\Sigma w}}\right),$ 
\item[(ii)] $\overline{I}(-\bar{\mu})\Phi\left(\frac{w^{\top}\bar{\mu}}{\sqrt{w^{\top}\Sigma w}}\right)\leq-\underline{I}(-\bar{\mu})\Phi\left(-\frac{w^{\top}\bar{\mu}}{\sqrt{w^{\top}\Sigma w}}\right),$ 
\end{itemize}
both of which are the same as \eqref{eq:condition-1} after further algebra, recalling that $\overline{I}(-\bar{\mu})=-\underline{I}(\bar{\mu})$ and $-\underline{I}(-\bar{\mu})=\overline{I}(\bar{\mu})$ by Lemma \ref{lemma:b1}, $w=\Sigma^{-1}\bar{\mu}$, and $\wnorm=\sqrt{\bar{\mu}^{\top}\Sigma^{-1}\bar{\mu}}$.\end{proof}

\begin{lemma}\label{lemma:a1-2-1} Consider the prior $\pi^{*}$
such that $\mu\sim\operatorname{unif}(\{-\bar{\mu},\bar{\mu}\})$,
and when $\mu=\bar{\mu}$, 
\[
U(\theta)=\begin{cases}
\overline{I}(\bar{\mu}), & \text{w.p. }p_{1},\\
\underline{I}(\bar{\mu}), & \text{w.p. }1-p_{1},
\end{cases}
\]
where $p_{1}$ is such that $p_{1}\overline{I}(\bar{\mu})+(1-p_{1})\underline{I}(\bar{\mu})=0$,
and when $\mu=-\bar{\mu}$, 
\[
U(\theta)=\begin{cases}
\overline{I}(-\bar{\mu}), & \text{w.p. }p_{2},\\
\underline{I}(-\bar{\mu}), & \text{w.p. }1-p_{2},
\end{cases}
\]
where $p_{2}$ is such that $p_{2}\overline{I}(-\bar{\mu})+(1-p_{2})\underline{I}(-\bar{\mu})=0$.
Given this prior, any decision rule is Bayes optimal under \eqref{eq:condition-large}. \end{lemma}
\begin{proof}
As \eqref{eq:condition-large} holds, we have $\underline{I}(\bar{\mu})<0<\overline{I}(\bar{\mu})$. Analogous to Lemma \ref{lemma:a1-1}, a Bayes rule must solve 
\begin{align*}
\min_{a\in[0,1]}\quad & \frac{1}{2}\cdot f(Y\mid \bar{\mu})\left[p_{1}\overline{I}(\bar{\mu})(1-a)+(1-p_{1})(-\underline{I}(\bar{\mu}))a\right]\\
+ & \frac{1}{2} \cdot f(Y\mid-\bar{\mu})\left[p_{2}\overline{I}(-\bar{\mu})(1-a)+(1-p_{2})(-\underline{I}(-\bar{\mu}))a\right].
\end{align*}
Since $p_{1}\overline{I}(\bar{\mu})+(1-p_{1})\underline{I}(\bar{\mu})=0$
and $p_{2}\overline{I}(-\bar{\mu})+(1-p_{2})\underline{I}(-\bar{\mu})=0$,
the objective is constant in $a$, hence the claim. 
\end{proof}
\begin{lemma}\label{lemma:a1-2-2} Consider the prior $\pi^{*}$
such that $\mu\sim\operatorname{unif}(\{-\bar{\mu},\bar{\mu}\})$,
when $\mu=\bar{\mu}$, 
\[
U(\theta)=\begin{cases}
\overline{I}(\bar{\mu}), & \text{w.p. }p_{1},\\
\underline{I}(\bar{\mu}), & \text{w.p. }1-p_{1},
\end{cases}
\]
where $p_{1}$ is such that $p_{1}\overline{I}(\bar{\mu})+(1-p_{1})\underline{I}(\bar{\mu})=0$,
and when $\mu=-\bar{\mu}$, 
\[
U(\theta)=\begin{cases}
\overline{I}(-\bar{\mu}), & \text{w.p. }p_{2},\\
\underline{I}(-\bar{\mu}), & \text{w.p. }1-p_{2},
\end{cases}
\]
where $p_{2}$ is such that $p_{2}\overline{I}(-\bar{\mu})+(1-p_{2})\underline{I}(-\bar{\mu})=0$.
Then, when \eqref{eq:condition-large} holds, we have
\[
\sup_{\pi\in\Gamma}r(d,\pi)=r(d,\pi^*)
\]
for any decision rule $d\in\mathcal{D}_{n}$ such that \eqref{pf:thm.1.2}
is true.
\end{lemma}
\begin{proof}
Note again that \eqref{eq:condition-large} implies $\underline{I}(\bar{\mu})<0$ and $\overline{I}(\bar{\mu})>0$. For any decision rule $d\in\mathcal{D}_{n}$ such that \eqref{pf:thm.1.2}
is true, algebra shows
\[
\overline{I}(\bar{\mu})(1-\mathbb{E}_{\bar{\mu}}[d])=-\underline{I}(\bar{\mu})\mathbb{E}_{\bar{\mu}}[d]=-\frac{\overline{I}(\bar{\mu})\underline{I}(\bar{\mu})}{\overline{I}(\bar{\mu})-\underline{I}(\bar{\mu})},
\]
\[
\overline{I}(-\bar{\mu})(1-\mathbb{E}_{-\bar{\mu}}[d])=-\underline{I}(-\bar{\mu})\mathbb{E}_{-\bar{\mu}}[d]=-\frac{\overline{I}(\bar{\mu})\underline{I}(\bar{\mu})}{\overline{I}(\bar{\mu})-\underline{I}(\bar{\mu})},
\]
implying $\sup_{\pi\in\Gamma}r(d,\pi)=-\frac{\overline{I}(\bar{\mu})\underline{I}(\bar{\mu})}{\overline{I}(\bar{\mu})-\underline{I}(\bar{\mu})}$.
Meanwhile, for any $d\in\mathcal{D}_{n}$ such that \eqref{pf:thm.1.2} is true, we may also calculate the following that yields the desired conclusion:
\[
r(d,\pi^*)=\frac{1}{2}p_{1}\overline{I}(\bar{\mu})+\frac{1}{2}p_{2}\overline{I}(-\bar{\mu})=-\frac{\overline{I}(\bar{\mu})\underline{I}(\bar{\mu})}{\overline{I}(\bar{\mu})-\underline{I}(\bar{\mu})}=\sup_{\pi\in\Gamma}r(d,\pi).
\]\end{proof}

\begin{lemma}\label{lemma:a1-3-1} In case (\ref{eq:condition-large}),
the function 
\begin{align*}
g(c) & =\frac{1}{2}\max\left\{ \overline{I}(\bar{\mu})\Phi\left(\frac{-\wnorm ^{2}+c}{\wnorm }\right),-\underline{I}(\bar{\mu})\Phi\left(\frac{\wnorm ^{2}-c}{\wnorm }\right)\right\} \\
 & +\frac{1}{2}\max\left\{ -\underline{I}(\bar{\mu})\Phi\left(\frac{\wnorm ^{2}+c}{\wnorm }\right),\overline{I}(\bar{\mu})\Phi\left(\frac{-\wnorm ^{2}-c}{\wnorm }\right)\right\} 
\end{align*}
is decreasing in $[0,c^{*}]$ and increasing in $[c^{*},\infty)$. \end{lemma} 
\begin{proof}
Note when $c\in[0,c^{*})$ and (\ref{eq:condition-large}) holds,
we have $-\underline{I}(\bar{\mu})>0$, 
\begin{align*}
-\underline{I}(\bar{\mu})\Phi\left(\frac{\wnorm ^{2}-c}{\wnorm }\right) & >\overline{I}(\bar{\mu})\Phi\left(\frac{-\wnorm ^{2}+c}{\wnorm }\right),\\
-\underline{I}(\bar{\mu})\Phi\left(\frac{\wnorm ^{2}+c}{\wnorm }\right) & >\overline{I}(\bar{\mu})\Phi\left(\frac{-\wnorm ^{2}-c}{\wnorm }\right),
\end{align*}
and therefore, 
\[
g(c)=-\frac{\underline{I}(\bar{\mu})}{2}\left[\Phi\left(\frac{\wnorm ^{2}-c}{\wnorm }\right)+\Phi\left(\frac{\wnorm ^{2}+c}{\wnorm }\right)\right],
\]
and 
\[
\frac{\partial g(c)}{\partial c}=-\frac{\underline{I}(\bar{\mu})}{2\wnorm }\left[\phi\left(\frac{\wnorm ^{2}+c}{\wnorm }\right)-\phi\left(\frac{\wnorm ^{2}-c}{\wnorm }\right)\right]<0.
\]
When $c=c^{*}$, note $\overline{I}(\bar{\mu})\Phi\left(\frac{-\wnorm ^{2}+c^{*}}{\wnorm }\right)=-\underline{I}(\bar{\mu})\Phi\left(\frac{\wnorm ^{2}-c^{*}}{\wnorm }\right)$
and 
\[
g(c^{*})=-\frac{\underline{I}(\bar{\mu})}{2}\left[\Phi\left(\frac{\wnorm ^{2}-c^{*}}{\wnorm }\right)+\Phi\left(\frac{\wnorm ^{2}+c^{*}}{\wnorm }\right)\right].
\]
When $c\in(c^{*},\infty)$, note 
\begin{align*}
-\underline{I}(\bar{\mu})\Phi\left(\frac{\wnorm ^{2}-c}{\wnorm }\right) & <\overline{I}(\bar{\mu})\Phi\left(\frac{-\wnorm ^{2}+c}{\wnorm }\right),\\
-\underline{I}(\bar{\mu})\Phi\left(\frac{\wnorm ^{2}+c}{\wnorm }\right) & >\overline{I}(\bar{\mu})\Phi\left(\frac{-\wnorm ^{2}-c}{\wnorm }\right),
\end{align*}
so that 
\begin{align*}
g(c) & =\frac{1}{2}\left\{ \overline{I}(\bar{\mu})\Phi\left(\frac{-\wnorm ^{2}+c}{\wnorm }\right)-\underline{I}(\bar{\mu})\Phi\left(\frac{\wnorm ^{2}+c}{\wnorm }\right)\right\} ,
\end{align*}
and 
\[
\frac{\partial g(c)}{\partial c}=\frac{1}{2\wnorm }\left\{ \overline{I}(\bar{\mu})\phi\left(\frac{-\wnorm ^{2}+c}{\wnorm }\right)-\underline{I}(\bar{\mu})\phi\left(\frac{\wnorm ^{2}+c}{\wnorm }\right)\right\} >0.
\]
Having shown that $\frac{\partial g(c)}{\partial c}<0$ when $c\in[0,c^{*})$
and $\frac{\partial g(c)}{\partial c}>0$ when $c\in(c^{*},\infty)$,
and since $g(c)$ is continuous at $c=c^{*}$, we conclude that $g$
is decreasing in $[0,c^{*}]$ and increasing $[c^{*},\infty)$.
\end{proof}
\begin{lemma}\label{lemma:a1-4-1}

In case (\ref{eq:condition-large}) and when $n>1$, let 
\[
d_{w_{t},0}(Y)=\mathbf{1}\{w_{t}^{\top}Y\geq0\},\text{ where }w_{t}=\Sigma^{-1}(t\bar{\mu}+(1-t)\dot{\mu}),\dot{\mu}\neq\mathbf{0},
\]
and $\dot{\mu}^{\top}\Sigma^{-1}\bar{\mu}=0$. Then, $\mathbb{E}_{\bar{\mu}}[d_{w_{t^{*}},0}(Y)]=\frac{\overline{I}(\bar{\mu})}{\overline{I}(\bar{\mu})-\underline{I}(\bar{\mu})}$,
where $t^{*}$ is defined in \eqref{eq:t.star}.

\end{lemma}
\begin{proof}
Write $f(t):=\mathbb{E}_{\bar{\mu}}[d_{w_{t},c}(Y)]=\Phi\left(\frac{w_{t}^{\top}\bar{\mu}}{\sqrt{w_{t}^{\top}\Sigma w_{t}}}\right)$ and  $k^{*}:=\frac{\overline{I}(\bar{\mu})}{\overline{I}(\bar{\mu})-\underline{I}(\bar{\mu})}$. Then, it suffices to show that $f(t^{*})=k^{*}$, which is equivalent to showing
\[
\frac{\left(w_{t^{*}}^{\top}\bar{\mu}\right)^{2}}{w_{t^{*}}^{\top}\Sigma w_{t^{*}}}  =\left(\Phi^{-1}\left(k^{*}\right)\right)^{2}.
\]
Furthermore, note
\begin{align*}
w_{t^{*}}^{\top}\bar{\mu} & =(t^{*}\bar{\mu}^{\top}\Sigma^{-1}+(1-t^{*})\dot{\mu}^{\top}\Sigma^{-1})\bar{\mu}
  =t^{*}\bar{\mu}^{\top}\Sigma^{-1}\bar{\mu}+(1-t^{*})\underbrace{\dot{\mu}^{\top}\Sigma^{-1}\bar{\mu}}_{=0}
  =t^{*}\wnorm^{2}
\end{align*}
and
\begin{align*}
w_{t^{*}}^{\top}\Sigma w_{t^{*}} & =(t^{*}\bar{\mu}^{\top}\Sigma^{-1}+(1-t^{*})\dot{\mu}^{\top}\Sigma^{-1})\Sigma(t^{*}\Sigma^{-1}\bar{\mu}+(1-t^{*})\Sigma^{-1}\dot{\mu})\\
 & =(t^{*})^{2}\bar{\mu}^{\top}\Sigma^{-1}\bar{\mu}+2t^{*}(1-t^{*})\underbrace{\bar{\mu}^{\top}\Sigma^{-1}\dot{\mu}}_{=0}+(1-t^{*})^{2}\underbrace{\dot{\mu}^{\top}\Sigma^{-1}\dot{\mu}}_{:=\wdotnorm^{2}}
 & =(t^{*})^{2}\wnorm^{2}+(1-t^{*})^{2}\wdotnorm^{2}.
\end{align*}
Therefore, substituting in these expressions, $t^{*}$ should be such
that 
\begin{alignat*}{2}
  \frac{\left(\Phi^{-1}(k^{*})\right)^{2}}{\wnorm^{2}}  =\frac{(t^{*})^{2}\wnorm^{2}}{(t^{*})^{2}\wnorm^{2}+(1-t^{*})^{2}\wdotnorm^{2}}.
\end{alignat*}
As $s^{*}:=\frac{\left(\Phi^{-1}(k^{*})\right)^{2}}{\wnorm^{2}}\in(0,1)$
due to (\ref{eq:condition-large}), we can solve the above equation for
$t^{*}$,  and after lengthy algebra,  find that
\begin{equation*}
t^{*} =\frac{1}{1\pm\sqrt{\frac{(1-s^{*})}{s^{*}}\cdot\frac{\wnorm^{2}}{\wdotnorm^{2}}}}.
\end{equation*}
\end{proof}

\section{Technical Lemmas Supporting Ex-post Analysis}

\begin{lemma}\label{lemma:a2-1} Suppose all the conditions of Theorem
\ref{thm:1} hold. If $\underline{I}(\priorlocation)<0<\overline{I}(\priorlocation)$,
then the unique $\Gamma$-PER optimal rule is 
\[
d_{\operatorname{PER}}^{*}(Y)=\begin{cases}
\frac{-\underline{I}(\priorlocation)}{\overline{I}(\priorlocation)-\underline{I}(\priorlocation)}, & \text{if }w^{\top}Y<0,\\
\frac{\overline{I}(\priorlocation)}{\overline{I}(\priorlocation)-\underline{I}(\priorlocation)}, & \text{if }w^{\top}Y\geq0.
\end{cases}
\]
Otherwise, $d_{w,0}^{*}$ is $\Gamma$-PER optimal. \end{lemma}
\begin{proof}
Let $f(Y\mid\mu)$ be the likelihood of $Y$. Then, for each action $a\in[0,1]$, we have

\begin{align*}
 & \sup_{\pi\in\Gamma}\int_{\tilde{\theta}\in\Theta}L(a,\tilde{\theta})d\pi_{\theta \mid Y}(\tilde{\theta})\\
\propto & \frac{1}{2}\left(\sup_{\pi_{\theta\mid\bar{\mu}}}\int_{\tilde{\theta}\in\Theta}U(\tilde{\theta})(\mathbf{1}\{U(\tilde{\theta})\geq0\}-a)d\pi_{\theta\mid\bar{\mu}}(\tilde{\theta})\right)f(Y\mid\bar{\mu})\\
+ & \frac{1}{2}\left(\sup_{\pi_{\theta\mid-\bar{\mu}}}\int_{\tilde{\theta}\in\Theta}U(\tilde{\theta})(\mathbf{1}\{U(\tilde{\theta})\geq0\}-a)d\pi_{\theta\mid-\bar{\mu}}(\tilde{\theta})\right)f(Y\mid-\bar{\mu})\\
= & \frac{1}{2}f(Y\mid\bar{\mu})\max\left\{ \overline{I}(\bar{\mu})(1-a),-\underline{I}(\bar{\mu})a\right\} \\
+ & \frac{1}{2}f(Y\mid-\bar{\mu})\max\left\{ \overline{I}(-\bar{\mu})(1-a),-\underline{I}(-\bar{\mu})a\right\} .
\end{align*}
Therefore, when randomization is allowed, the $\Gamma$-PER optimal rule can be found by solving  
\begin{align}
 & \inf_{a\in[0,1]}V_{\Gamma}(a),\label{eq:pf:per.1}
\end{align}
where 
\begin{align*}
V_{\Gamma}(a) & :=f(Y\mid\bar{\mu})\max\left\{ \overline{I}(\bar{\mu})(1-a),-\underline{I}(\bar{\mu})a\right\} 
  +f(Y\mid-\bar{\mu})\max\left\{ \overline{I}(\bar{\mu})a,-\underline{I}(\bar{\mu})(1-a)\right\},
\end{align*}
since $\overline{I}(-\bar{\mu})=-\underline{I}(\bar{\mu})$
and $-\underline{I}(-\bar{\mu})=\overline{I}(\bar{\mu})$
by Lemma \ref{lemma:b1}. Note that, for each action $a\in[0,1]$, 
\begin{align*}
\overline{I}(\bar{\mu})(1-a)\geq-\underline{I}(\bar{\mu})a & \Leftrightarrow a\leq\frac{\overline{I}(\bar{\mu})}{\overline{I}(\bar{\mu})-\underline{I}(\bar{\mu})},\\
\overline{I}(\bar{\mu})a\geq-\underline{I}(\bar{\mu})(1-a) & \Leftrightarrow a\geq\frac{-\underline{I}(\bar{\mu})}{\overline{I}(\bar{\mu})-\underline{I}(\bar{\mu})},
\end{align*}
where $\overline{I}(\bar{\mu})-\underline{I}(\bar{\mu})\geq0$ and $\overline{I}(\bar{\mu})>0$ due to our normalization $\overline{I}(\bar{\mu})+\underline{I}(\bar{\mu})>0$.
The conclusion follows from discussing the following two cases.

\paragraph{Case 1: $\underline{I}(\bar{\mu})\protect\geq0$.}

In this case, $\frac{-\underline{I}(\bar{\mu})}{\overline{I}(\bar{\mu})-\underline{I}(\bar{\mu})} \leq a\leq\frac{\overline{I}(\bar{\mu})}{\overline{I}(\bar{\mu})-\underline{I}(\bar{\mu})}$ for all $a\in[0,1]$. Thus, (\ref{eq:pf:per.1}) becomes 
\begin{align*}
 & \inf_{a\in[0,1]}\underset{>0}{\underbrace{\overline{I}(\bar{\mu})}}\left\{ f(Y\mid\bar{\mu})+\left[f(Y\mid-\bar{\mu})-f(Y\mid\bar{\mu})\right]a\right\} .
\end{align*}
The $\Gamma$-PER optimal rule is $\mathbb{\mathbf{1}}\left\{ f(Y\mid-\bar{\mu})-f(Y\mid\bar{\mu})\leq0\right\} $,
which is $d_{w,0}^{*}$ after further algebra.

\paragraph{Case 2: $\underline{I}(\bar{\mu})<0<\overline{I}(\bar{\mu})$.}

In this case, we have (again, note $\overline{I}(\bar{\mu})+\underline{I}(\bar{\mu})>0$):
\[
0<\frac{-\underline{I}(\bar{\mu})}{\overline{I}(\bar{\mu})-\underline{I}(\bar{\mu})}<\frac{\overline{I}(\bar{\mu})}{\overline{I}(\bar{\mu})-\underline{I}(\bar{\mu})}<1.
\]
Therefore: 
\begin{itemize}
\item When $a\in \left[0,\frac{-\underline{I}(\bar{\mu})}{\overline{I}(\bar{\mu})-\underline{I}(\bar{\mu})}\right)$, 
\begin{align*}
V_{\Gamma}(a)= & f(Y\mid\bar{\mu})\overline{I}(\bar{\mu})(1-a)-f(Y\mid-\bar{\mu})\underline{I}(\bar{\mu})(1-a)\\
= & \underset{>0}{\underbrace{\left\{ f(Y\mid\bar{\mu})\overline{I}(\bar{\mu})+f(Y\mid-\bar{\mu})\left(-\underline{I}(\bar{\mu})\right)\right\} }}(1-a).
\end{align*}
Thus, $V_{\Gamma}(\cdotp)$ is strictly decreasing in $a\in\left[0,\frac{-\underline{I}(\bar{\mu})}{\overline{I}(\bar{\mu})-\underline{I}(\bar{\mu})}\right)$. 
\item When $a\in\left[\frac{-\underline{I}(\bar{\mu})}{\overline{I}(\bar{\mu})-\underline{I}(\bar{\mu})},\frac{\overline{I}(\bar{\mu})}{\overline{I}(\bar{\mu})-\underline{I}(\bar{\mu})}\right),$
\begin{align*}
V_{\Gamma}(a)= & f(Y\mid\bar{\mu})\overline{I}(\bar{\mu})(1-a)+f(Y\mid-\bar{\mu})\overline{I}(\bar{\mu})a\\
= & f(Y\mid\bar{\mu})\overline{I}(\bar{\mu})+\underset{>0}{\underbrace{\overline{I}(\bar{\mu})}}\left(f(Y\mid-\bar{\mu})-f(Y\mid\bar{\mu})\right)a.
\end{align*}
Thus, if $f(Y\mid-\bar{\mu})>f(Y\mid\bar{\mu})$, $V_{\Gamma}(\cdotp)$
is strictly increasing in $a$, and if $f(Y\mid-\bar{\mu})<f(Y\mid\bar{\mu})$,
$V_{\Gamma}(\cdotp)$ is strictly decreasing in $a$. When $f(Y\mid-\bar{\mu})=f(Y\mid\bar{\mu})$,
$V_{\Gamma}(\cdotp)$ is constant in $a$. 
\item When $a\in\left[\frac{\overline{I}(\bar{\mu})}{\overline{I}(\bar{\mu})-\underline{I}(\bar{\mu})},1\right]$,
\begin{align*}
V_{\Gamma}(a) & =\underset{>0}{\underbrace{(f(Y\mid-\bar{\mu})\overline{I}(\bar{\mu})-f(Y\mid\bar{\mu})\underline{I}(\bar{\mu}))}}~a,
\end{align*}
implying that $V_{\Gamma}(\cdotp)$ is strictly increasing in $a$. 
\end{itemize}
In sum, if $f(Y\mid-\bar{\mu})>f(Y\mid\bar{\mu})$,
$V_{\Gamma}(\cdotp)$ is first strictly decreasing and then increasing,
with the minimum achieved at $a=\frac{-\underline{I}(\bar{\mu})}{\overline{I}(\bar{\mu})-\underline{I}(\bar{\mu})}$;
if $f(Y\mid-\bar{\mu})<f(Y\mid\bar{\mu})$, $V_{\Gamma}(\cdotp)$
is first strictly decreasing and increasing with the minimum achieved
at $a=\frac{\overline{I}(\bar{\mu})}{\overline{I}(\bar{\mu})-\underline{I}(\bar{\mu})}$;
If $f(Y\mid-\bar{\mu})=f(Y\mid\bar{\mu})$ (which only happens
in a null set), $V_{\Gamma}(\cdotp)$ achieves minimum at any point
between $[\frac{-\underline{I}(\bar{\mu})}{\overline{I}(\bar{\mu})-\underline{I}(\bar{\mu})},\frac{\overline{I}(\bar{\mu})}{\overline{I}(\bar{\mu})-\underline{I}(\bar{\mu})}]$.
Therefore, the unique $\Gamma$-PER optimal rule is 
\[
d(Y)=\begin{cases}
\frac{\overline{I}(\bar{\mu})}{\overline{I}(\bar{\mu})-\underline{I}(\bar{\mu})}, & \text{if }f(Y\mid-\bar{\mu})\leq f(Y\mid\bar{\mu})\Longleftrightarrow\bar{\mu}^{\top}\Sigma^{-1}Y\geq0,\\
\frac{-\underline{I}(\bar{\mu})}{\overline{I}(\bar{\mu})-\underline{I}(\bar{\mu})} & \text{if }f(Y\mid-\bar{\mu})>f(Y\mid\bar{\mu})\Longleftrightarrow\bar{\mu}^{\top}\Sigma^{-1}Y<0.
\end{cases}
\]
\end{proof}
\begin{lemma}\label{lemma:a2-2} Suppose all conditions of Theorem
\ref{thm:1} hold true. Then, $d_{w,0}^{*}$ is always the $\Gamma$-PER
optimal non-randomized rule. \end{lemma}
\begin{proof}
When randomization is not allowed, we need to solve (by a similar
derivation to Lemma \ref{lemma:a2-1})  
$\inf_{a\in\{0,1\}}V_{\Gamma}(a)$,
where $V_{\Gamma}(a)$ is defined in (\ref{eq:pf:per.1}). As 
\begin{align*}
V_{\Gamma}(1) & =-f(Y\mid\bar{\mu})\underline{I}(\bar{\mu})+f(Y\mid-\bar{\mu})\overline{I}(\bar{\mu}),\\
V_{\Gamma}(0) & =f(Y\mid\bar{\mu})\overline{I}(\bar{\mu})-f(Y\mid-\bar{\mu})\underline{I}(\bar{\mu}),
\end{align*}
the optimal action would be $a=1$ if and only if $V_{\Gamma}(1)\leq V_{\Gamma}(0)$, which is equivalent to 
$d_{w,0}^{*}$ after further algebra. 
\end{proof}

\section{Additional Results}\label{sec:addtional.results}

\begin{lemma}\label{lemma:b1} Consider a treatment choice problem
with welfare function \eqref{eq:welfare}, statistical model \eqref{eq:normal_model}
and a set of priors \eqref{eq:Gamma}, that satisfies Assumption \ref{asm:1}.
Then, the following statements hold: 
\begin{itemize}
\item[(i)] $\overline{I}(-\bar{\mu})=-\underline{I}(\bar{\mu})$,
$\underline{I}(-\bar{\mu})=-\overline{I}(\bar{\mu})$; 
\item[(ii)] $\overline{I}(-\bar{\mu})+\underline{I}(-\bar{\mu})=-(\underline{I}(\bar{\mu})+\overline{I}(\bar{\mu}))$. 
\end{itemize}
\end{lemma}
\begin{proof}
\textbf{Statement (i).} Here, we only show that $\overline{I}(-\bar{\mu})=-\underline{I}(\bar{\mu})$.
An analogous argument proves $\underline{I}(-\bar{\mu})=-\overline{I}(\bar{\mu})$.
By definition, 
\[
\overline{I}(-\bar{\mu})=\sup_{\{\theta\in\Theta:m(\theta)=-\bar{\mu}\}}U(\theta).
\]
Since $\Theta$ is centrosymmetric and both $U(\cdotp)$ and $m(\cdot)$
are linear, we have 
\begin{align*}
\sup_{\{\theta\in\Theta:m(\theta)=-\bar{\mu}\}}U(\theta) & =\sup_{\{-\theta\in\Theta:-m(-\theta)=-\bar{\mu}\}}-U(-\theta)\\
 & =\sup_{\{\tilde{\theta}\in\Theta:-m(\tilde{\theta})=-\bar{\mu}\}}-U(\tilde{\theta})\\
 & =-\left\{ \inf_{\{\tilde{\theta}\in\Theta:m(\tilde{\theta})=\bar{\mu}\}}U(\tilde{\theta})\right\} \\
 & =-\underline{I}(\bar{\mu})
\end{align*}
as required.

\textbf{Statement (ii).} Directly follows from summing both equalities
in statement (i).
\end{proof}

\begin{lemma}\label{lem:profiled.regret}
In Example \ref{ex:stoye}, suppose $k>\sqrt{\frac{\pi}{2}}\sigma$.
Then, if $\bar{\mu}>0$ is sufficiently small, the corresponding $\Gamma$-PER
optimal rule is dominated in terms of profiled regret.
\end{lemma}

\begin{proof}
Denote by $d_{MMR,\text{linear}}^{*}$ the least randomizing 
global MMR optimal rule derived in \citet{MQS}. We aim to show that when
$k>\sqrt{\frac{\pi}{2}}\sigma$ and $\bar{\mu}>0$ is sufficiently
small, the associated $\Gamma$-PER optimal rule $d^*_{\text{PER}}$ is such that
\begin{equation}
\bar{R}(d_{\operatorname{PER}}^{*},\mu)\geq\bar{R}(d_{MMR,\text{linear}}^{*},\mu),\text{ for all \ensuremath{\mu}}\geq0,\label{pf:profiled.regret.1}
\end{equation}
with the inequality strict for all $\mu>0$. A symmetry argument then
immediately implies that $d_{\operatorname{PER}}^{*}$ is dominated in terms of profiled regret.

\paragraph{Step 1} Pick any $0<\bar{\mu}<\sqrt{\frac{\pi}{2}}\sigma$. We show
that $\bar{R}(d_{\operatorname{PER}}^{*},\mu)=(\mu+k)\left(1-\mathbb{E}_{\mu}\left[d_{\operatorname{PER}}^{*}(\hat{\mu})\right]\right)$ for all $\mu\geq0$. By results in \citet[Appendix B.3.1, ][]{MQS}, we may write the profiled
regret of a rule $d$ as
\begin{align*}
 & \bar{R}(d,\mu)\\
= & \begin{cases}
(-\mu+k)\mathbb{E}_{\mu}\left[d(\hat{\mu})\right], & \text{if }\mu<-k,\\
\max\left\{ (\mu+k)\left(1-\mathbb{E}_{\mu}\left[d(\hat{\mu})\right]\right),(-\mu+k)\mathbb{E}_{\mu}\left[d(\hat{\mu})\right]\right\} , & \text{if }-k\leq\mu\leq k,\\
(\mu+k)(1-\mathbb{E}_{\mu}\left[d(\hat{\mu})\right]), & \text{if }\mu>k.
\end{cases}
\end{align*}

Thus, it suffices to show 
\begin{align*}
 & \max\left\{ (\mu+k)\left(1-\mathbb{E}_{\mu}\left[d_{\operatorname{PER}}^{*}(\hat{\mu})\right]\right),(-\mu+k)\mathbb{E}_{\mu}\left[d_{\operatorname{PER}}^{*}(\hat{\mu})\right]\right\} \\
= & (\mu+k)\left(1-\mathbb{E}_{\mu}\left[d_{\operatorname{PER}}^{*}(\hat{\mu})\right]\right)
\end{align*}
for any $0\leq\mu\leq k$.  As $\sqrt{\frac{\pi}{2}}\sigma<k$, Corollary \ref{coro:2} implies
that the corresponding $\Gamma$-PER rule for any $0<\bar{\mu}<\sqrt{\frac{\pi}{2}}\sigma$
is 
\begin{align*}
d_{\operatorname{PER}}^{*}(\hat{\mu}) & =\frac{k+\overline{\mu}}{2k}\mathbf{1}\left\{ \hat{\mu}\geq0\right\} +\frac{k-\overline{\mu}}{2k}\mathbf{1}\left\{ \hat{\mu}<0\right\} \\
 & =\frac{1}{2}+\frac{\overline{\mu}}{2k}\left(\mathbf{1}\left\{ \hat{\mu}\geq0\right\} -\mathbf{1}\left\{ \hat{\mu}<0\right\} \right),
\end{align*}
and further algebra shows $\mathbb{E}_{\mu}\left[d_{\operatorname{PER}}^{*}(\hat{\mu})\right]=\frac{1}{2}+\frac{\overline{\mu}}{2k}\left(1-2\Phi\left(-\frac{\mu}{\sigma}\right)\right)$.
Then,
\[
(\mu+k)\left(1-\mathbb{E}_{\mu}\left[d_{\operatorname{PER}}^{*}(\hat{\mu})\right]\right)\geq(-\mu+k)\mathbb{E}_{\mu}\left[d_{\operatorname{PER}}^{*}(\hat{\mu})\right]
\]
 if and only if 
\[
(\mu+k)\left(\frac{1}{2}-\frac{\overline{\mu}}{2k}\left(1-2\Phi\left(-\frac{\mu}{\sigma}\right)\right)\right)\geq(-\mu+k)\left(\frac{1}{2}+\frac{\overline{\mu}}{2k}\left(1-2\Phi\left(-\frac{\mu}{\sigma}\right)\right)\right),
\]
which is equivalent to 
\begin{equation}
\frac{\mu}{\left(1-2\Phi\left(-\frac{\mu}{\sigma}\right)\right)}\geq\overline{\mu}.\label{pf:profiled.regret.2}
\end{equation}
Note the left hand side of (\ref{pf:profiled.regret.2}) is increasing
in $\mu$, and $\lim_{\mu\downarrow0}\left(\frac{\mu}{\left(1-2\Phi\left(-\frac{\mu}{\sigma}\right)\right)}\right)=\frac{\sigma}{2\phi\left(0\right)}=\sigma\sqrt{\frac{\pi}{2}}>\overline{\mu}$.
As $0<\bar{\mu}<\sqrt{\frac{\pi}{2}}\sigma$, we conclude that (\ref{pf:profiled.regret.2})
indeed holds, and it follows 
\begin{equation}\label{pf:profiled.regret.3}
\bar{R}(d_{\operatorname{PER}}^{*},\mu)=(\mu+k)\left(1-\mathbb{E}_{\mu}\left[d_{\operatorname{PER}}^{*}(\hat{\mu})\right]\right),
\end{equation}
for all $\mu\geq0$. As a result, we conclude that, for any $0<\bar{\mu}<\sqrt{\frac{\pi}{2}}\sigma$,
we have 
\[
\bar{R}(d_{\operatorname{PER}}^{*},\mu)=(\mu+k)\left(\frac{1}{2}-\frac{\overline{\mu}}{2k}\left(1-2\Phi\left(-\frac{\mu}{\sigma}\right)\right)\right),\text{ for all }\mu\geq0.
\]

\paragraph{Step 2} We show that $\bar{R}(d_{\operatorname{PER}}^{*},\mu)$ is
strictly increasing in $\mu$ for any $0<\bar{\mu}<\sqrt{\frac{\pi}{2}}\sigma$
small enough. To see this, note 
\begin{align*}
\frac{\partial\bar{R}(d_{\operatorname{PER}}^{*},\mu)}{\partial\mu} & =\frac{1}{2}-\frac{\overline{\mu}}{2k}\left(1-2\Phi\left(-\frac{\mu}{\sigma}\right)\right)\\
 & -(\mu+k)\frac{\overline{\mu}}{k}\frac{1}{\sigma}\phi\left(\frac{\mu}{\sigma}\right)\\
 & =\frac{1}{2}-\frac{\overline{\mu}}{2k}\left(1-2\Phi\left(-\frac{\mu}{\sigma}\right)+2(\mu+k)\frac{1}{\sigma}\phi\left(\frac{\mu}{\sigma}\right)\right),
\end{align*}
in which note (1): $0<1-2\Phi\left(-\frac{\mu}{\sigma}\right)\leq1$
for all $\mu\geq0$, and $(2)$: $2(\mu+k)\frac{1}{\sigma}\phi\left(\frac{\mu}{\sigma}\right)>0$
is first increasing and then decreasing with a unique and finite maximum
$C_{k}>0$. Therefore, we conclude that for any $0<\overline{\mu}<\min\left(\frac{k}{\left(1+C_{k}\right)},\sqrt{\frac{\pi}{2}}\sigma\right)$,
we have $\frac{\partial\bar{R}(d_{\operatorname{PER}}^{*},\mu)}{\partial\mu}>0$
for all $\mu\geq0$. That is, $\bar{R}(d_{\operatorname{PER}}^{*},\mu)$
is strictly increasing in $\mu\in[0,\infty]$.

\paragraph{Step 3} \citet{MQS} have already shown that \[\sup_{\mu\geq0}\bar{R}(d_{MMR,\text{linear}}^{*},\mu)=\bar{R}(d_{MMR,\text{linear}}^{*},0)=\frac{k}{2}.\]
Then, by the conclusion from step 3 and noting $\bar{R}(d_{\operatorname{PER}}^{*},0)=\frac{k}{2}$,
we see that (\ref{pf:profiled.regret.1}) indeed holds with the inequality strict for all $\mu>0$, completing
the proof.
\end{proof}

\section{General Nature of Our Main Results}\label{sec:generality}

In the main text, we derived finite-sample $\Gamma$-MMR and -PER
optimal rules for a class of priors such that $\pi_{\mu}\sim\text{unif\ensuremath{\left(\left\{  -\overline{\mu},\overline{\mu}\right\}  \right)}}$.
In this section, we discuss the implications of our results when $\pi_{\mu}$
has a more general structure (e.g., with a continuous support). 

It is relatively straightforward to extend our $\Gamma$-PER results
(Theorem \ref{thm:2}) to other forms of $\pi_{\mu}$. Let 
\[
V_{\Gamma}(a):=V_{\Gamma}(a,Y):=\sup_{\pi\in\Gamma}\int_{\tilde{\theta}\in\Theta}L(a,\tilde{\theta})d_{\theta\mid Y}(\tilde{\theta}).
\]
The ex-post $\Gamma$-PER criterion aims to solve $\min_{a\in[0,1]}V_{\Gamma}(a).$
Given a general $\pi_{\mu}$ and the normal likelihood of $Y$, we
can derive the the posterior distribution of the reduced-form parameter
$\mu$ given $Y$, written as $\pi_{\mu\mid Y}$, using the standard
Bayes rule. In light of the structure of the class of priors $\Gamma$,
we see 
\begin{align*}
V_{\Gamma}(a)= & \int\sup_{\mu^{*}\in I(x)}\left\{ \mu^{*}\left[\mathbf{1}\left\{ \mu^{*}\geq0\right\} -a\right]\right\} d\pi_{\mu\mid Y}(x)\\
= & \int\max\left\{ \overline{I}(x)\left(1-a\right),-\underline{I}(x)a\right\} d\pi_{\mu\mid Y}(x).
\end{align*}
Note $\overline{I}(x)\left(1-a\right)\geq-\underline{I}(x)a$ if and
only if $\overline{I}(x)\geq a\left(\overline{I}(x)-\underline{I}(x)\right)$,
equivalent to (assuming $\overline{I}(x)-\underline{I}(x)>0$ for
all $x$ in the support of $\pi_{\mu\mid Y}$) $a\leq\frac{\overline{I}(x)}{\overline{I}(x)-\underline{I}(x)}.$
Thus, for each $a\in[0,1]$, we have
\begin{align*}
V_{\Gamma}(a) & =\int\left[\mathbf{1}\left\{ a\leq\frac{\overline{I}(x)}{\overline{I}(x)-\underline{I}(x)}\right\} \overline{I}(x)\left(1-a\right)-\underline{I}(x)a\mathbf{1}\left\{ a>\frac{\overline{I}(x)}{\overline{I}(x)-\underline{I}(x)}\right\} \right]d\pi_{\mu\mid Y}(x)\\
 & =\left(1-a\right)\int_{\left\{ x\in\mathbb{R}^{n}:\frac{\overline{I}(x)}{\overline{I}(x)-\underline{I}(x)}\geq a\right\} }\overline{I}(x)d\pi_{\mu\mid Y}(x)-a\int_{\left\{ x\in\mathbb{R}^{n}:\frac{\overline{I}(x)}{\overline{I}(x)-\underline{I}(x)}<a\right\} }\underline{I}(x)d\pi_{\mu\mid Y}(x).
\end{align*}
The key observation is that the value of $V_{\Gamma}(a)$ for each
$a\in[0,1]$ can be numerically solved, although it may not have a
closed-form solution. Therefore, in general, we can still find $d_{\text{PER}}^{*}(Y)$
by numerically solving $\min_{a\in[0,1]}V_{\Gamma}(a)$. Note as the solution of $\min_{a\in[0,1]}V_{\Gamma}(a)$ may or may not be at the corners $\{0,1\}$, it is in general \emph{not} the case that the associated $\Gamma$-PER rule always randomizes.
 
In contrast, the $\Gamma$-MMR optimal rule is more difficult to
find once we allow general forms of $\pi_{\mu}$. Similar to the unconstrained
MMR optimality problem, one often needs to resort to the ``guess-and-verify''
strategy by forming a least favorable prior (LPF) --- unlike the unconstrained
scenario, for the $\Gamma$-MMR problem, the LFP has to come from
$\Gamma$, which restricts Nature's strategy significantly. The particular
form of $\pi_{\mu}\sim\text{unif\ensuremath{\left(\left\{  -\overline{\mu},\overline{\mu}\right\}  \right)}}$
makes our guess of the LFP much easier, while for other cases, it
may be more difficult. That said, although the exact analytic form
of $\Gamma$-MMR optimal rule is difficult to acquire for general
$\pi_{\mu}$, one may utilize recent computational techniques \citep{fernandez2024epsilon,guggenberger2025numerical} to numerically
find an $\varepsilon$-$\Gamma$-MMR optimal rule with near optimal
convergence property. Moreover, we think that the main qualitative
insights regarding $\Gamma$-MMR optimal rules may extend to other
forms of $\pi_{\mu}$. For example, we conjecture that whenever the
identification power of the model is sufficiently large compared to
the informativeness of the data (in a sense to be suitably defined),
$\Gamma$-MMR optimal rule will take a form of a non-randomized threshold
rule; otherwise, the optimal rule may be randomized. We leave the
verification of these conjectures for future research.

\end{document}